\documentclass{article}
\usepackage[utf8]{inputenc}

\usepackage[margin=1.2in]{geometry}
\usepackage{amsmath}
\usepackage{amssymb}
\usepackage{amsthm}
\usepackage{algorithm}
\usepackage{algpseudocode}
\usepackage{hyperref}
\usepackage{graphicx}
\usepackage{subcaption}
\usepackage[font={small}]{caption}
\usepackage{thm-restate}
\usepackage[normalem]{ulem}
\usepackage{framed}
\usepackage{complexity}
\usepackage{enumitem}

\usepackage[dvipsnames]{xcolor}

\hypersetup{
pagebackref,
colorlinks=true,
urlcolor=blue,
linkcolor=blue,
citecolor=OliveGreen,
}

\newcommand{\nc}{\newcommand}
\newcommand{\DMO}{\DeclareMathOperator}

\algnotext{EndFor}
\algnotext{EndIf}
\algnotext{EndWhile}
\algnotext{EndProcedure}

\allowdisplaybreaks

\nc{\MS}{\mathcal{S}}
\nc{\MR}{\mathcal{R}}
\nc{\cM}{\mathcal{M}}

\nc{\MZ}{\mathcal{Z}}

\DMO{\Binom}{Binom}
\renewcommand{\E}{\mathbb{E}}
\DMO{\Var}{Var}

\newcommand{\N}{\mathbb{N}}

\nc{\BN}{\mathbb{N}}

\nc{\BZ}{\mathbb{Z}}

\newcommand{\eps}{\varepsilon}
\nc{\ep}{\eps}

\newcommand{\cF}{\mathcal{F}}
\newcommand{\cE}{\mathcal{E}}
\newcommand{\cG}{\mathcal{G}}
\newcommand{\cV}{\mathcal{V}}
\newcommand{\cA}{\mathcal{A}}

\newcommand{\bA}{\mathbb{A}}
\newcommand{\bB}{\mathbb{B}}
\newcommand{\calS}{\mathcal{S}}

\DeclareMathOperator{\den}{den}
\DeclareMathOperator{\mindeg}{min-deg}
\renewcommand{\varepsilon}{\epsilon}
\let\Pr=\relax
\DeclareMathOperator*{\Pr}{\mathbf{Pr}}

\nc{\rgp}{\mathrm{RP}}

\newtheorem{theorem}{Theorem}
\newtheorem*{theorem*}{Theorem}

\newtheorem{observation}[theorem]{Observation}
\newtheorem{lemma}[theorem]{Lemma}
\newtheorem*{lemma*}{Lemma}
\newtheorem*{obs*}{Observation}

\newtheorem{corollary}[theorem]{Corollary}

\newtheorem{hypothesis}[theorem]{Hypothesis}

\newcommand\MYcurrentlabel{xxx}

\newcommand{\MYstore}[2]{%
  \global\expandafter \def \csname MYMEMORY #1 \endcsname{#2}%
}

\newcommand{\MYload}[1]{%
  \csname MYMEMORY #1 \endcsname%
}

\newcommand{\MYnewlabel}[1]{%
  \renewcommand\MYcurrentlabel{#1}%
  \MYoldlabel{#1}%
}

\newcommand{\MYdummylabel}[1]{}

\newcommand{\torestate}[1]{%
  \let\MYoldlabel\label%
  \let\label\MYnewlabel%
  #1%
  \MYstore{\MYcurrentlabel}{#1}%
  \let\label\MYoldlabel%
}
\newcommand{\restatelemma}[1]{%
  \let\MYoldlabel\label
  \let\label\MYdummylabel
  \begin{lemma*}[Restatement of Lemma~\ref{#1}]
    \MYload{#1}
  \end{lemma*}
  \let\label\MYoldlabel
}

\newcommand{\restatetheorem}[1]{%
  \let\MYoldlabel\label
  \let\label\MYdummylabel
  \begin{theorem*}[Restatement of Theorem~\ref{#1}]
    \MYload{#1}
  \end{theorem*}
  \let\label\MYoldlabel
}

\title{The Strongish Planted Clique Hypothesis and Its Consequences}

\author{Pasin Manurangsi\\ Google Research\\ \texttt{pasin@google.com}
\and Aviad Rubinstein\\ Stanford University\\ \texttt{aviad@cs.stanford.edu}
\and Tselil Schramm\\ Stanford University\\ \texttt{tselil@stanford.edu}}

\date{\today}

\begin{document}

\maketitle

\begin{abstract}
We formulate a new hardness assumption, the {\em Strongish Planted Clique Hypothesis (SPCH)}, which postulates that any algorithm for planted clique must run in time $n^{\Omega(\log{n})}$ (so that the state-of-the-art running time of $n^{O(\log n)}$ is optimal up to a constant in the exponent).

We provide two sets of applications of the new hypothesis. 
First, we show that SPCH implies (nearly) tight inapproximability results for the following well-studied problems in terms of the parameter $k$: Densest $k$-Subgraph, Smallest $k$-Edge Subgraph, Densest $k$-Subhypergraph, Steiner $k$-Forest, and Directed Steiner Network with $k$ terminal pairs. For example, we show, under SPCH, that no polynomial time algorithm achieves $o(k)$-approximation for Densest $k$-Subgraph. This inapproximability ratio improves upon the previous best $k^{o(1)}$ factor from (Chalermsook et al., FOCS 2017). Furthermore, our lower bounds hold even against fixed-parameter tractable algorithms with parameter $k$.

Our second application focuses on the complexity of graph pattern detection. For both induced and non-induced graph pattern detection, we prove hardness results under SPCH, which improves the running time lower bounds obtained by (Dalirrooyfard et al., STOC 2019) under the Exponential Time Hypothesis. 
\end{abstract}

\setcounter{page}{0}
\thispagestyle{empty}
\newpage

\section{Introduction}
The last couple of decades have seen dramatic advances in our understanding of parameterized, fine-grained, and average-case complexity. 
To a large extent, this progress has been enabled by bolder computational hardness assumptions, beyond the classical $\P \neq \NP$. 
Two notable assumptions in these fields are the Exponential Time Hypothesis and the Planted Clique Hypothesis.  
In this paper we propose a new hypothesis, the {\em Strongish Planted Clique Hypothesis}, which strengthens the Planted Clique Hypothesis in the style of the Exponential Time Hypothesis.
We show that this hypothesis has interesting implications in the parameterized complexity of both approximation problems and graph pattern detection.

\paragraph{Exponential Time Hypothesis}
The Exponential Time Hypothesis (ETH)~\cite{IP01} is a pessimistic version of $\P \neq \NP$ which postulates that solving 3-SAT on $n$ variables requires time $2^{\Omega(n)}$. In other words, the running times of the state-of-the-art (and the brute-force) algorithms for 3-SAT are optimal up to a constant factor in the exponent. 
ETH has important applications in parameterized complexity (e.g.~\cite{CJ03,LMS11,Lin15,CPP16,DVVW19}) and hardness of approximation (e.g.~\cite{BKW15,BravermanKRW17,Rubinstein17,Manurangsi17,MR17,DinurM18}). 
In the past several years, further progress was achieved by assuming stronger variants of ETH such as the Strong ETH (SETH) in fine-grained complexity~\cite{IPZ01,Williams18}, Gap-ETH in parameterized complexity~\cite{ChalermsookCKLM17,BGS18,LRSZ17}, and ETH for~\PPAD~in Algorithmic Game Theory~\cite{BPR16,Rubinstein16}. 

\paragraph{Planed Clique Hypothesis}
In the Planted $\kappa$-Clique Problem, the goal is to distinguish (with high probability) between graphs sampled from one of the following distributions: Uniformly at random%
\footnote{I.e. from the {\em Erd\H{o}s-R\'{e}nyi} (ER) distribution over $n$-vertex graphs where each edge appears independently with probability $1/2$.}; and
uniformly at random, with an added $\kappa$-clique.
While statistically it is easy to distinguish the two distributions for $\kappa$ as little as $2.1\log(n)$, the {\em Planted Clique Hypothesis} (PCH) postulates that no polynomial time algorithm can solve this problem, even for $\kappa$ as large as $o(\sqrt{n})$.
The history of this problem goes back to Karp~\cite{Karp76} and Jerrum~\cite{Jerrum92}, and in the past decade it has been popular as a hardness assumption for both worst-case~\cite{ABBG10,HK11,alon2011inapproximability,BBBCT13,BCKS16} and average-case~\cite{BR13,HWX15,WBS16,GMZ17,BBH18,BB19} problems. 
A simple $n^{\Theta(\log(n))}$-time algorithm for the planted-$\kappa$-clique problem non-deterministically guesses $\ell=\Theta(\log(n))$ vertices from the clique, and then checks whether all of their common neighbors form a clique.
There are several other algorithms that also solve this problem in time $n^{\Theta(\log(n))}$~\cite{FS97,FK03,MM15,Barman18}, but no faster algorithm is known for $\kappa = O(n^{0.49})$.

\subsection*{Strongish Planted Clique Hypothesis}
In analogy to the Exponential Time Hypothesis for 3-SAT, we propose the following hypothesis, which postulates that the state-of-the-art algorithms for the Planted Clique Problem are optimal up to a constant factor in the exponent. 
A {\em Strong} Planted Clique Hypothesis, in analogy with SETH, would specify a precise constant in the exponent---our hypothesis is merely Strong-ish.
We let $\cG(n,p)$ denote the Erd\H{o}s-R\'{e}nyi distribution with parameter $p$, and $\cG(n,p,\kappa)$ denote the Erd\H{o}s-R\'{e}nyi distribution with a planted $\kappa$-clique. 

\begin{hypothesis}[Strongish Planted Clique Hypothesis (SPCH)] \label{hyp:strong-planted-clique}
There exists a constant $\delta \in (0,\tfrac{1}{2})$ such that no $n^{o(\log n)}$-time algorithm $\cA$ satisfies both of the following:
\begin{itemize}
\item (Completeness) $\Pr_{G \sim \cG(n, \frac{1}{2}, \lceil n^{\delta} \rceil)}[\cA(G) = 1] \geq 2/3$.
\item (Soundness) $\Pr_{G \sim \cG(n, \frac{1}{2})}[\cA(G) = 1] \leq 1/3$.
\end{itemize}
\end{hypothesis}
In addition to the lack of algorithmic progress toward refuting this hypothesis, we note that $n^{\Theta(\log(n))}$ is in fact provably optimal for the Sum-of-Squares hierarchy~\cite{BarakHKKMP19}, which captures the state-of-the-art algorithmic techniques for a number of average-case problems. It is also known to be tight for statistical algorithms~\cite{FGRVX17,brennan2020statistical}.

\subsection{Our Contributions: Hardness results from Strongish Planted Clique Hypothesis}

Our main technical contributions are in exploring the implications of our new SPCH in parameterized complexity.
We prove two types of hardness results: hardness of approximation, and hardness of (exact) graph pattern detection. Due to the rich literature for each problem we consider, we will only mention the most relevant results here and defer a more comprehensive discussion to Section~\ref{sec:other-prior-work}.

\subsubsection{Hardness of approximation from SPCH}

At the heart of our work is the study of the \emph{Densest $k$-Subgraph} problem, in which we are given an undirected graph $G = (V, E)$ and a positive integer $k$. The goal is to output a subset $S \subseteq V$ of $k$ vertices that induces as many edges as possible. There is a trivial $O(k)$-approximation for the problem: return an arbitrary set of $\lfloor k / 2 \rfloor$ edges. We show that assuming SPCH, this algorithm is optimal:
\begin{theorem}
\torestate{
\label{thm:dks-inapprox}
Assuming the Strongish Planted Clique Hypothesis (Hypothesis~\ref{hyp:strong-planted-clique}), there is no $f(k) \cdot \poly(n)$-time algorithm that can approximate Densest $k$-Subgraph on $n$-vertex graphs to within a factor $o(k)$ for any function $f$. Furthermore, this holds even in the \emph{perfect completeness} case where the input graph is promised to contain a $k$-clique.
}
\end{theorem}
\noindent Theorem~\ref{thm:dks-inapprox} improves upon the inapproximability ratio of $k^{o(1)}$ shown in~\cite{ChalermsookCKLM17} under Gap-ETH.

The approximability of Densest $k$-Subgraph is known to be intimately related to that of numerous other problems. As such, our tight hardness of approximation for Densest $k$-Subgraph immediately implies several tight approximability results as corollaries, which we list below.

\begin{itemize}

\item {\em Smallest $k$-Edge Subgraph}: given an undirected graph $G = (V, E)$ and a positive integer $k$, find a smallest subset $S \subseteq V$ that induces at least $k$ edges. For this problem, the trivial solution that chooses $k$ edges arbitrarily is an $O(\sqrt{k})$-approximation since even the optimum requires at least $\sqrt{k}$ vertices. We show that this is tight (Corollary~\ref{cor:smallest-subgraph}): no fixed-parameter tractable (FPT) (in $k$) algorithm can achieve $o(\sqrt{k})$ approximation ratio.

\item {\em Steiner $k$-Forest} (aka {\em $k$-Forest}): given an edge-weighted undirected graph $G = (V, E)$, a set $\{(s_1, t_1), \dots, (s_\ell, t_\ell)\}$ of demand pairs and a positive integer $k$, the goal is to find a (not necessarily induced) subgraph of $G$ with smallest total edge weight that connects at least $k$ demand pairs. In Corollary~\ref{cor:k-forest}, we show that no FPT (in $k$) algorithm can achieve $o(\sqrt{k})$ approximation ratio. This matches the $O(\sqrt{k})$-approximation algorithm by Gupta et al.~\cite{GuptaHNR10}.

\item {\em Directed Steiner Network} (aka {\em Directed Steiner Forest}): given an edge-weighted directed graph $G = (V, E)$ and a set $\{(s_1, t_1), \dots, (s_k, t_k)\}$ of $k$ demand pairs, the goal is to find a (not necessarily induced) subgraph of $G$ with smallest total edge weight in which $t_i$ is reachable from $s_i$ for all $i = 1, \dots, k$. We prove that no FPT (in $k$) algorithm achieves $o(\sqrt{k})$-approximation (Corollary~\ref{cor:dsn}). This nearly matches the best known polynomial algorithms by Chekuri et al.~\cite{ChekuriEGS11} and Feldman et al.~\cite{FeldmanKN12}, both of which achieve a $O(k^{1/2 + \eps})$-approximation for any constant $\eps > 0$. Our bound improves upon a $k^{1/4 - o(1)}$ ratio from~\cite{DinurM18} under Gap-ETH.

\item {\em Densest $k$-Subhypergraph}: given a hypergraph $G = (V, E)$ and a positive integer $k$, output a $k$-size subset $S \subseteq V$ that maximizes the number of hyperedges fully contained in $S$. The trivial algorithm that outputs any hyperedge (of size at most $k$) obtains a $2^k$-approximation. We prove a matching lower bound (Theorem~\ref{thm:dksh}): no $2^{o(k)}$-approximation FPT (in $k$) algorithm exists. This resolves an open question posed by Cygan et al.~\cite{cygan-dksh}, assuming SPCH.

\end{itemize}

\subsubsection{Hardness of graph pattern detection from SPCH}
{\em Graph Pattern Detection}, also known as {\em Subgraph Isomorphism} and closely related to {\em Motif Discovery}, is a fundamental problem in graph algorithms: Given a host graph $G$ and a pattern graph $H$, decide whether $G$ contains a subgraph $S$ isomorphic to  $H$. There are two main variants of this problem, where $S$ is either required to be {\em induced} or {\em not-necessarily-induced}. For both variants, there is a brute-force $n^{O(k)}$-time algorithm. 
We prove matching SPCH-hardness for both variants. In addition to beating the ETH-based state-of-the-art for these results, we highlight that our reductions to graph pattern detection problems are extremely simple (in contrast to the prior work).

For induced subgraph detection, we prove the following: 
\begin{theorem}\torestate{ \label{thm:induced-pattern}
Assuming the Strongish Planted Clique Hypothesis (Hypothesis~\ref{hyp:strong-planted-clique}), for every $k$-node pattern $H$, there is no algorithm that solves the induced pattern detection problem on $n$-vertex graphs in time $f(k)\cdot n^{o(k)}$ for any function $f$.}
\end{theorem}
\noindent Our $n^{\Omega(k)}$ lower bound for all patterns can be compared to recent work of~\cite{DVVW19}, who proved: (i) $n^{\Omega(\log(k))}$ lower bound for every pattern assuming ETH; (ii) $n^{\Omega(\sqrt{k})}$ lower bound for every pattern assuming ETH and the Hadwiger conjecture; and  (iii) $n^{\Omega(k/\log(k))}$ lower bound for most patterns.%
\footnote{In the sense of a pattern sampled randomly from $\cG(k,1/2)$.}

For  not-necessarily-induced subgraph detection, it is no longer true that every pattern is hard (e.g.~it is trivial to find a not-necessarily-induced subgraph isomorphic to an independent set). But we prove (Corollary~\ref{cor:not-induced}) that for most patterns%
\footnote{in particular any pattern with a constant fraction of the $k\choose{2}$ possible edges.} detection requires $n^{\Omega(k)}$ time assuming SPCH. For comparison,~\cite{DVVW19} proved that under ETH, not-necessarily-induced subgraph detection requires $n^{\Omega (\omega(H))}$ time, where $\omega(H)$ denotes the clique number of the pattern $H$. (Note that $\omega(H) = \Theta(\log k)$ for most patterns.)

\paragraph{$k$-biclique detection}
For the special case where the pattern is a $k$-biclique, our aforementioned $n^{\Omega(k)}$ hardness for non-induced subgraph detection rules out even constant factor approximations (Corollary~\ref{cor:biclique}). This improves over $n^{\Omega(\sqrt{k})}$ lower bounds under ETH for the exact case~\cite{Lin15} or Gap-ETH for approximation~\cite{ChalermsookCKLM17}.

\paragraph{Densest $k$-Subgraph}
We obtain our pattern detection result by first showing a $n^{\Omega(k)}$ running time bound for $O(1)$-approximating Densest $k$-Subgraph (Theorem~\ref{thm:dks-time}).
This improves upon the previous lower bound who give a running time lower bound of $n^{\Omega(\log k)}$ assuming Gap-ETH~\cite{ChalermsookCKLM17}. The aforementioned lower bounds for pattern detection follow almost trivially from our running time lower bound for Densest $k$-Subgraph; see Sections~\ref{subsec:biclique} and~\ref{subsec:pattern-detection} for more detail.

\subsection{Techniques}

The starting point for all of our reductions is a {\em randomized graph product}: starting with an instance $G$ of the planted clique problem on $n$ vertices and any integers $\ell \le n$ and $N$, we produce a graph $G'$ by taking its vertices to be $N$ randomly sampled subsets $S_1,\ldots,S_N$ of $\ell$ vertices each, and we add an edge on $S_i,S_j$ if and only if their union induces a clique in $G$.

The randomized graph product~\cite{BermanS89} (and its derandomized variant~\cite{AlonFWZ95}) has a long history in proving hardness of approximating Maximum Clique. While not stated explicitly, it was also used to prove parameterized inapproximability of Densest $k$-Subgraph in~\cite{ChalermsookCKLM17}. As we will explain in more detail below, the main difference between our proof and previous works lie in the soundness, where we appeal to the fact that $G \sim \cG(n, 1/2)$ to achieve tighter bounds. 

Since we would like to prove hardness of approximating Densest $k$-Subgraph in the perfect completeness case, our goal is to show (for appropriately chosen values of $N, \ell$) that
\begin{itemize}
\item (Completeness) if $G$ contains a large clique, then with high probability so does $G'$, and,
\item (Soundness) if $G$ is random, then $G'$ does not have small dense subgraphs (with high probability).
\end{itemize}

For the completeness, if the starting graph $G$ has a $\kappa$-clique, then the set of $S_i$ that fall entirely within the $\kappa$-clique will form a clique in $G'$ (the expected size is $(\frac{\kappa}{n})^\ell\cdot N$). This part of the proof is exactly the same as that in the aforementioned previous works.

To prove soundness, we calculate the probability that a specific $\gamma k$-dense $k$-subgraph appears in $G'$, then take a union bound over all $\le \binom{N}{k}\cdot 2^{\binom{k}{2}}$ possible such subgraphs.
Our simple argument hinges on showing that a $k$-subgraph with $\gamma k^2$ edges in $G'$ induces (with high probability) at least $\Omega(\gamma k^2 \ell^2)$ edges in $G$, and since any such set of edges appears in $G$ with probability at most $2^{-\Omega(\gamma k^2 \ell^2)}$, by choosing $\ell$ sufficiently large we can beat this union bound. 
To argue that small subgraphs with $m$ edges in $G'$ induce small subgraphs with $\Omega(m \ell^2)$ edges in $G$, we use that the randomly chosen $S_i$ (for an appropriate choice of $N \ll n^\ell$ and $k_*$ sufficiently small) form a {\em disperser}: the union of {\em any} $t \le k_*$ of the $S_i$ contains $\Omega(t\ell)$ vertices of $G$ with high probability.\footnote{The fact that the randomized graph product yields a disperser has been used in previous works as well, see e.g. \cite{BermanS89,zuckerman1996unapproximable,ChalermsookCKLM17}. }
This implies that for $k \le k_*$, each $k$-subgraph of $G'$ corresponds to a union of $k$ {\em pairwise} mostly-non-overlapping subsets of $\ell$ vertices.
Now, since each edge between mostly non-overlapping sets in $G'$ corresponds to a $\Omega(\ell)$-clique in $G$, this in turn can be used to show that any $k$-subgraph of $G'$ with density $\gamma k$ corresponds to a subgraph of $G$ with $\Omega(\gamma k^2 \ell^2)$ edges. 
In this way we rule out the existence of $\gamma k$-dense $k$-subgraphs in $G'$ (with high probability).

By carefully choosing the parameters $N$, $\ell$, $\gamma$, to control the completeness, soundness, and reduction size, we get a fine-grained reduction from Planted $\kappa$-Clique to Densest $k$-Subgraph.
Our results for other problems are obtained via direct reductions from the Densest $k$-Subgraph problem.

We end by stressing that our new soundness proof gives a strong quantitative improvement over prior results, which is what enables us to achieve $k^{\Omega(1)}$ inapproximability. 
Specifically, the previous soundness proof from~\cite{ChalermsookCKLM17}---in turn adapted from~\cite{Manurangsi17}---relies on showing that the graph is $t$-biclique-free for some $t \in \N$; they then apply the so-called Kovari-Sos-Turan theorem~\cite{kHovari1954problem} to deduce that any $k$-subgraph contains at most $O(k^{2 - 1/t})$ edges. Notice that this  gap is only $O(k^{1/t})$, and $t$ cannot be a constant as otherwise the completeness and soundness case can be distinguished in time $n^{O(t)} = \poly(n)$. As a result, their technique cannot yield an $k^{\Omega(1)}$-factor inapproximability for Densest $k$-Subgraph.
Similarly, the techniques from~\cite{ChalermsookCKLM17} {cannot} give a running time lower bound of the form $n^{\omega(\log k)}$ for $O(1)$-approximation of Densest $k$-Subgraph.
The reason is that, to get a constant gap bounded from one, they need to select $t = O(\log k)$. Once again, this is in contrast to our technique which yields a tight running time lower bound of $n^{\Omega(k)}$ in this setting, although our proof requires a different starting hardness result (from SPCH).

\subsection{Discussion and Open Questions}

In this work, we have proposed the Strongish Planted Clique Hypothesis, and used it prove several tight hardness of approximation or running time lower bound results. One direction for future investigation is to use SPCH to derive other interesting lower bounds. 

Another intriguing question directly related to our hardness of approximation results is whether we can strengthen the inapproximability factors from $k^{\Omega(1)}$ to $n^{\Omega(1)}$. Whether it is hard to approximate Densest $k$-Subgraph to within a $n^{\Omega(1)}$ factor is a well-known open problem and is related to some other conjectures, such as the Sliding-Scaling Conjecture~\cite{BellareGLR93}\footnote{Specifically, from a reduction of~\cite{CharikarHK11}, $n^{\Omega(1)}$-factor hardness of approximation of Densest $k$-Subgraph also implies that of the Label Cover problem.}. As such, it would be interesting if one can prove this hardness under some plausible assumption. We remark that an attempt has been made in this direction, using the so-called ``Log-Density Threshold'' approach~\cite{BCCFV10}, which posits a heuristic for predicting which average-case Densest-$k$-Subgraph problems are hard.

The approach has also been applied to other related questions~\cite{ChlamtacDK12,ChlamtacDM17,ChlamtacMMV17}. Nonetheless, there is still little evidence that these average-case DkS problems are indeed hard; not even lower bounds against strong SDP relaxations are known, although there are some matching lower bounds against LP hierarchies~\cite{BhaskaraCVGZ12,ChlamtacM18}. 

Finally, it is of course interesting to either refute or find more evidence supporting the Strongish Planted Clique Hypothesis. As stated earlier, the current best supporting evidence is the Sum-of-Squares lower bound from~\cite{BarakHKKMP19}. Can such a lower bound be extended to, e.g., rule out any semi-definite programs of size $n^{o(\log n)}$ (\`a la~\cite{LeeRS15} for CSPs)?

\subsection{Other Related Works}
\label{sec:other-prior-work}
Historically, postulating hardness for average-case problems has been helpful in illuminating the landscape for hardness of approximation, beginning with Feige's seminal random-$3$-SAT hypothesis~\cite{Feige02} and its numerous consequences (e.g.~\cite{DanielyS16}). See also \cite{Daniely16,alekhnovich2011more,applebaum2010public,barak2013optimality}.

As discussed briefly above, the Planted Clique Hypothesis (PCH), which states that there is no polynomial-time algorithm for planted clique, has many known consequences for hardness of approximation.
We draw attention in particular to the work of~\cite{alon2011inapproximability}, which also obtains hardness of approximation results based on PCH.
But even assuming the SPCH, their results can only rule out $n^{\polylog(\kappa)}$-time algorithms for $2^{\log (\kappa)^{2/3}}$-approximating densest-$\kappa$-subgraph for $\kappa = n^{\Omega(1)}$.
Their reduction also uses a graph product, but 
the set of vertices of their new graph $G'$ contains all $\ell$-size subsets of vertices of $G$. In contrast, we employ the {\em randomized} graph product, where we only randomly pick some $\ell$-size subsets---this allows us to better control the instance size blowup, which turns out to be crucial for obtaining our tight inapproximability and running time results.

\medskip
Below we discuss in more detail the previous works for each of the problems studied here.

\paragraph{Densest $k$-Subgraph.} 
The problem is well-studied in the approximation algorithms literature (e.g.~\cite{FPK01,KP93,BCCFV10}). The best known polynomial time algorithm~\cite{BCCFV10} achieves an approximation ratio of $O(n^{1/4 + \eps})$ for any $\eps > 0$.
Even though the NP-hardness of Densest $k$-Subgraph follows immediately from that of $k$-Clique~\cite{Karp72}, no NP-hardness of approximation of Densest $k$-Subgraph even for a small factor of 1.001 is known. Nonetheless, several hardness of approximation results are known under stronger assumptions~\cite{Feige02,Khot06,RaghavendraS10,alon2011inapproximability,barak2013optimality,BravermanKRW17,Manurangsi17}. Among these, only~\cite{alon2011inapproximability,barak2013optimality} and~\cite{Manurangsi17} yield super-constant inapproximability ratios. Specifically,~\cite{alon2011inapproximability} rules out $2^{O((\log n)^{2/3})}$-approximation in polynomial time under a hypothesis similar to SPCH, \cite{Manurangsi17} rules out $n^{o(1 / \poly\log\log n)}$-approximation under ETH, and
\cite{barak2013optimality} rules out $n^{O(1)}$ approximations under a strong conjecture regarding the optimality of semidefinite programs for solving every random CSP.

Our hardness result holds even in the so-called \emph{perfect completeness} case, where we are promised that the graph $G$ contains a $k$-clique. In this case, a quasi-polynomial time approximation scheme exists~\cite{FS97}. Braverman et al.~\cite{BravermanKRW17} showed that this is tight: there exists a constant $\eps > 0$ for which $(1 + \eps)$-approximation of Densest $k$-Subgraph in the perfect completeness case requires $n^{\tilde{\Omega}(\log n)}$-time assuming ETH. We remark that the hardness from~\cite{Manurangsi17} also applies in the perfect completeness case, but it only rules out polynomial time algorithms.

For the parameterized version of the problem, its W[1]-hardness follows immediately from that of $k$-Clique~\cite{DowneyF95-w1}. Chalermsook et al.~\cite{ChalermsookCKLM17} showed that no $k^{o(1)}$-approximation is achievable in FPT time, unless Gap-ETH is false. This hardness also holds in the perfect completeness case, and yields a running time lower bound of $n^{\Omega(\log k)}$ for any constant factor approximation.

\paragraph{$k$-Biclique.} Similar to Densest $k$-Subgraph, while the NP-hardness for the exact version of $k$-Biclique has long been known (e.g.~\cite{Joh87}), even 1.001-factor NP-hardness of approximation has not yet been proved although inapproximability results under stronger assumptions are known~\cite{Feige02,Khot06,BhangaleGHKK16,Man17-ICALP}. Specifically, under strengthened variants of the Unique Games Conjecture, the problem is hard to approximate to within a factor of $n^{1 - \eps}$ for any $\eps > 0$~\cite{BhangaleGHKK16,Man17-ICALP}.

On the parameterized complexity front, whether $k$-Biclique is FPT (in $k$) was a well-known open question (see e.g.~\cite{DowneyF13}). This was resolved by Lin~\cite{Lin15} who showed that the problem is W[1]-hard and further provided a running time lower bound of $n^{\Omega(\sqrt{k})}$ under ETH. As stated above, this running time lower bound was extended to rule out any constant factor approximation in~\cite{ChalermsookCKLM17} under Gap-ETH. Furthermore, \cite{ChalermsookCKLM17} showed that no $o(k)$-approximation exists in FPT time.

\paragraph{Smallest $k$-Edge Subgraph.} Most of the hardness of approximation for Densest $k$-Subgraph easily translates to Smallest $k$-Edge Subgraph as well, with a polynomial loss in the factor of approximation. For example, the hardness from~\cite{Manurangsi17} implies that Smallest $k$-Edge Subgraph cannot be approximated to within a factor of $n^{1/\poly\log\log n}$ in polynomial time, assuming ETH. On the other hand, Chlamtac et al.~\cite{ChlamtacDK12} devised an $O(n^{3 - 2\sqrt{2} + \eps})$-approximation algorithm for any constant $\eps > 0$ for the problem; this remains the best known approximation algorithm for the problem.

\paragraph{Densest $k$-Subhypergraph.} Apart from the hardness results inherited from Densest $k$-Subgraph, not much is known about Densest $k$-Subhypergraph. Specifically, the only new hardness is that of Applebaum~\cite{Applebaum13}, who showed that the problem is hard to approximate to within $n^{\eps}$ for some constant $\eps > 0$, assuming a certain cryptographic assumption; this holds even when each hyperedge has a constant size. On the other hand, the only (non-trivial) approximation algorithm is that of Chlamtac et al.~\cite{ChlamtacDKKR18} which achieves $O(n^{0.698})$-approximation when the hypergraph is 3-uniform.

\paragraph{Steiner $k$-Forest.} The Steiner $k$-Forest problem is a generalization of several well-known problems, including the Steiner Tree problem and the $k$-Minimum Spanning Tree problem. This problem was first explicitly defined in~\cite{HajiaghayiJ06} and subsequently studied in~\cite{SegevS10,GuptaHNR10,DinitzKN17}. In terms of the number of vertices $n$ of the input graph, the best known approximation ratio achievable in polynomial time is $O(\sqrt{n})$~\cite{GuptaHNR10} (assuming that $k \leq \poly(n)$); furthermore, when edge weights are uniform, a better approximation ratio of $O(n^{0.449})$ is achievable in polynomial time~\cite{DinitzKN17}. On the other hand, as stated earlier, in terms of $k$, the best known approximation ratio is $O(\sqrt{k})$~\cite{GuptaHNR10}.

A reduction in~\cite{HajiaghayiJ06} together with W[1]-hardness of $k$-Clique~\cite{DowneyF95-w1} implies that Steiner $k$-Forest is W[1]-hard with respect to the parameter $k$. We are not aware of any further parameterized complexity study of this problem (with respect to parameter $k$).

\paragraph{Directed Steiner Network.} Several polynomial time approximation algorithms have been proposed for the Directed Steiner Network problem~\cite{CCCDGGL99,ChekuriEGS11,FeldmanKN12,BBMRY13,CDKL17,AbboudB18}; in terms of the number of vertices $n$ of the input graph the best known approximation ratio is $O(n^{2/3 + \eps})$~\cite{BBMRY13} and in terms of $k$ the best known ratio is $O(k^{1/2 + \eps})$ for any constant $\eps > 0$~\cite{ChekuriEGS11,FeldmanKN12}. On the hardness front, Dodis and Khanna~\cite{DodisK99} shows that the problem is quasi-NP-hard to approximate to within a factor of $2^{(\log n)^{1 - \eps}}$ for any constant $\eps > 0$.

Guo et al.~\cite{GuoNS11} show that the exact version of the problem is W[1]-hard with respect to parameter $k$. Later,~\cite{DinurM18} rules out even $k^{1/4 - o(1)}$-approximation in FPT time, under Gap-ETH.

\paragraph{Graph Pattern Detection.} 
As discussed earlier,~\cite{DVVW19} give ETH-based hardness results for graph pattern detection, both in the induced and non-induced case.
The complexity for many special patterns has also been considered, e.g.~$k$-cliques, $k$-bicliques (mentioned above), and $k$-cycles (e.g.~\cite{AYZ97,YZ04,DKS17,DVVW19,LV20}). A $k$-clique can be detected in time $O(n^{\lceil k/3\rceil \omega})$ using fast matrix multiplication~\cite{NP85}, and the $k$-Clique Conjecture in Fine-Grained Complexity postulates that this is essentially optimal~\cite{Williams18}. Any other pattern over $k$ vertices can be detected in time $O(n^{k-1}$), without using fast matrix multiplication~\cite{BKS18}. There is also an extensive body of work on {\em counting} the number of occurrences of a pattern in a host graph (e.g.~\cite{Milo02,KLL13,UBK13,WW13,CM14,CDM17}).

\subsection*{Preliminaries and Notation}
For a natural number $n \in \N$, we use $[n]$ to denote the set of integers up to $n$, $[n] = \{1,\ldots,n\}$.
We will use the abbreviation ``w.h.p.'' to mean ``with high probability.''

For an undirected graph $G = (V, E)$, we use $\deg_G(v)$ to denote the degree of a vertex $v \in V$, and $\mindeg(G)$ to denote $\min_{v \in V(G)} \deg_G(v)$.
For a subset $S \subseteq V$, we use $G[S] = (V, E[S])$ to denote the induced subgraph of $G$ on subset $S$.

 The \emph{density} of $G$, denoted by $\den(G)$, is defined as $|E| / |V|$. We also use $\den_{\leq k}(G)$ to denote $\max_{S \subseteq V, |S| \leq k} \den(G[S])$, the maximum density of subgraphs of $G$ of at most $k$ vertices.

\section{Randomized Graph Product}
In this section we formally define our reduction, and analyze its soundness and completeness in terms of the reduction parameters (in later sections we instantiate the parameters differently for each target bound).
Our reduction takes as input a graph and applies the following {\em randomized graph product}:

\begin{figure}[h!]
\begin{framed}
\textbf{Randomized Graph Product~\cite{BermanS89}} \\
\textbf{Input: } $n$-vertex Graph $G = (V, E)$, positive integers $N, \ell$. \\
\textbf{Output: } Graph $G' = (V', E')$. \\
The graph $G'$ is constructed as follows.
\begin{enumerate}
\item For each $i \in [N]$, sample $S_i \subseteq V$ by independently sampling $\ell$ vertices uniformly from $V$.
\item Let $V' = \{S_1, \dots, S_N\}$.
\item For every distinct $i, j \in [N]$, include $(S_i, S_j)$ in $E'$ if and only if $S_i \cup S_j$ induces a clique in $G$.
\end{enumerate}
\end{framed}
\label{fig:my_label}
\end{figure}
\noindent We use $\rgp_{N,\ell}(G)$ to denote the distribution of outputs of the above reduction on input graph $G$.

We will show that for well-chosen $N$ and $\ell$, the randomized graph product  $\rgp_{N,\ell}$ reduces the planted $n^{\delta}$-clique problem to densest-$k$ subgraph for $k = k(N,\ell,\delta,n)$. 
That is, a sample from $\rgp_{N,\ell} \circ \cG(n,\frac{1}{2})$ has no dense $k$-subgraph with probability close to $1$, and if on the other hand $G$ is a graph with an $n^{\delta}$-clique, then a sample from $\rgp_{N,\ell}(G)$ has a dense $k$-subgraphs with probability close to $1$.

\subsection{Completeness}
We first prove that applying the randomized graph product to a graph with a large clique results in a graph with a large clique.

\begin{lemma}[Completeness]\label{lem:completeness}
Suppose that $\delta \in (0, 1), N, \ell, k \in \N$ are such that $N \geq 10 k \cdot n^{(1 - \delta)\ell}$ and $k \ge 20$. If $G$ contains a clique of size $\lceil n^{\delta} \rceil$, then
\[
\Pr_{G' \sim \rgp_{N,\ell}(G)}\left[\, G' \text{ contains a $k$-clique }\right]\ge 0.9.
\]
\end{lemma}

\begin{proof}
Let $C \subseteq V$ be the $\lceil n^{\delta} \rceil$-size clique in $G$. For each $i \in [N]$, $\Pr[S_i \subseteq C] = \left(\frac{\lceil n^{\delta} \rceil}{n}\right)^\ell \geq n^{-(1 - \delta) \ell}$. 
By our lower bound on $N$, the expected number of indices $i \in [N]$ such that $S_i \subseteq C$ is at least $10 k$, and thus a Chernoff bound implies that with probability $1-\exp(-4k) \ge 0.9$, there exists at least $k$ indices $i \in [N]$ such that $S_i \subseteq C$.
By definition of the randomized graph product $\rgp_{N,\ell}$, these subsets form a clique in $G'$. 
\end{proof}

\subsection{Soundness}
\label{subsec:generic-soundness}
We now prove that if we apply the randomized graph product to a graph drawn from $G(n,\frac{1}{2})$, with high probability the resulting graph has no small subgraphs which are too dense.

\begin{lemma}[Soundness] \label{lem:soundness-generic}
Suppose that $\delta \in (0, 1), N, \ell, k \in \N$ are such that $N \leq 1000k \cdot n^{(1 - \delta) \ell}, \ell \geq k \geq 20$ and $n^{0.99\delta} \geq k\ell$. If $G$ is drawn from $\cG(n, \frac{1}{2})$, then
\begin{align*}
\Pr_{\substack{G \sim \cG(n,\frac{1}{2})\\G' \sim \rgp_{N,\ell}(G)}}\left[\, \den_{\leq k}(G') \leq \frac{10^7 \log n}{\ell \delta^2}\,\right] \ge 0.9.
\end{align*}
\end{lemma}

We will use the following observation that allows us to translate a graph with large density to a (sub)graph with large minimum degree. This observation is folklore and appears implicitly e.g. in~\cite{Charikar00}. Nonetheless, we provide the proof in Appendix~\ref{app:soundness-helper-lemmas} for completeness.

\begin{observation} \label{obs:density-min-deg}
For any $H = (V_H, G_H)$, there exists $S' \subseteq H$ such that $\mindeg(H[S']) \geq \den(H)$.
\end{observation}

Another auxiliary lemma that is useful for us is that for any subset $M \subseteq [N]$ not too large, the size of the union $\left|\cup_{j \in M} S_j\right|$ is not too small relative to $|M| \cdot \ell$. This lemma is also standard and has been used in prior works (e.g.~\cite{BermanS89,zuckerman1996unapproximable,ChalermsookCKLM17}). Once again, we provide its proof in Appendix~\ref{app:soundness-helper-lemmas} for completeness.

\begin{lemma}\torestate{ \label{lem:disperser}
Suppose $N \le 1000 \ell n^{(1-\delta)\ell}$, $20\le \ell$. 
Then with probability at least 0.95 over a sample $G' \sim \rgp_{N,\ell} \circ \cG(n, \frac{1}{2})$, $G' = (\{S_i\}_{i \in [N]},E')$, the following event occurs: for every $M \subseteq [N]$ with $|M| \le n^{0.99\delta} / \ell$, we have $|\bigcup_{i \in M} S_i| \geq 0.01 \delta |M| \ell$.}
\end{lemma}

With the above observation and lemma ready, we can now prove the soundness guarantee.

\begin{proof}[Proof of Lemma~\ref{lem:soundness-generic}]
We will assume that the event in Lemma~\ref{lem:disperser} occurs and show that conditioned on this event, under the randomness of $G$, with probability 0.95 {\em all} $k$-subgraphs of $G'$ have density at most $d := 10^7 \log n / (\ell \delta^2)$. Lemma~\ref{lem:soundness-generic} then follows immediately since $(0.95)^2 > 0.9$.

Consider any $J \subseteq [N]$ of size $k' \leq k$.
For brevity, let $\cF(J)$ denote the set of all graphs whose vertices are $S_{j}$ for $j \in J$ and the minimum degree is at least $d$.
For each $F = (\{S_j\}_{j \in J}, E_F) \in \cF(J)$, let $\cE^G(F)$ denote $\bigcup_{\{S_j, S_{j'}\} \in E_F} \{\{u, v\} \mid u \in S_j, v \in S_{j'} \text{ and } u \ne v\}$, or in words, the set of edges of $G$ which have one endpoint in $S_j$ and one endpoint in $S_{j'}$ for an edge $(S_j, S_{j'})\in E_F$.
Observe that
\begin{align}
\Pr_G[G'[\{S_j\}_{j \in J}] = F] 
&\leq \Pr_G[\cE^G(F) \subseteq E]
\, =\, 2^{-|\cE^G(F)|}, \label{eq:union-original-random-graph}
\end{align}
where the inequality follows by definition of the randomized graph product --- since $(S_i,S_j)$ is an edge if and only if $S_i \cup S_j$ is a clique in $G$, the event $G'[\{S_j\}_{j \in J}] = F$ contains the event $\cE^G(F) \subseteq E$ --- and the final equality follows because $G \sim \cG(n, \frac{1}{2})$.
To bound $|\cE^G(F)|$, let $\cV^G(F) := \cup_{j \in J} S_j$.

Since we have conditioned on the event in Lemma~\ref{lem:disperser} occurring,\footnote{Note that $|J| = k' \leq k \leq n^{0.99\delta} / \ell$ by our assumption and hence $J$ satisfies the condition in Lemma~\ref{lem:disperser}.} we have $|\cV^G(F)| \geq 0.01\delta k' \ell$. Now, consider any $v \in \cV^G(F)$; from definition of $\cV^G(F)$, $v \in S_j$ for some $j \in J$. Since $F \in \cF(J)$, $S_j$ must have at least $d$ neighbors in $F$. 
Let $S_{j_1}, \dots, S_{j_{d'}}$ denote $S_j$'s neighbors in $F$, with $d' \ge d$. By applying the bound in Lemma~\ref{lem:disperser}, we have $|S_{j_1} \cup \dots \cup S_{j_{d'}}| \geq 0.01\delta d' \ell \geq 0.01\delta d \ell$. In other words, $v$ has degree at least $0.01 \delta d \ell - 1 \geq 0.005 \delta d \ell$ in the graph $(\cV^G(F), \cE^G(F))$. This implies that
\begin{align*}
|\cE^G(F)| &\geq \frac{1}{2} \left(0.01\delta k' \ell\right) \left(0.005 \delta d \ell\right) \\
&\geq 10^{-5} \delta^2 k' d \ell^2.
\end{align*}
Plugging the above bound back into~\eqref{eq:union-original-random-graph}, we have
\begin{align} \label{eq:fixed-subgraph-probability-bound}
\Pr_G[G'[\{S_j\}_{j \in J}] = F] \leq 2^{-10^{-5} \delta^2 k' d \ell^2}
\end{align}

We can use the above inequality to bound the probability that $\{S_j\}_{j \in J}$ induces a subgraph with minimum degree at least $d$ as follows:
\begin{align*}
\Pr[\mindeg(G'[\{S_j\}_{j \in J}]) \geq d]
\, =\, \sum_{F \in \cF(J)}  \Pr_G[G'[\{S_j\}_{j \in J}] = F] 
&\overset{\eqref{eq:fixed-subgraph-probability-bound}}{\leq} 2^{(k')^2} \cdot 2^{-10^{-5} \delta^2 k' d \ell^2}
 \, \leq\, 2^{-10^{-6} \delta^2 k' d \ell^2}.
\end{align*}
where the first inequality follows because there are at most $2^{(k')^2}$ subgraphs of an $k'$-vertex graph, and to obtain the final inequality we have applied that $\ell \geq k \ge k'$ and $d \geq 10^7 / (\delta^2 \ell)$.

Applying Observation~\ref{obs:density-min-deg}, the existence of a $k$-subgraph with density at least $d$ would imply the existence of some $J \subseteq [N]$ with $|J| \le k$ and minimum degree at least $d$. 
Taking union bound over all $J \subseteq [N]$ of size at most $k$ and applying our above bound, we have
\begin{align*}
\Pr[\den_{\leq k}(G)] 
\, \leq\, \sum_{k'=1}^{k} N^{k'} \cdot 2^{-10^{-6} \delta^2 k' d \ell^2} 
&= \sum_{k'=1}^{k} \left(N \cdot 2^{-10^{-6} \delta^2 d \ell^2}\right)^{k'} \\
 &\leq \sum_{k'=1}^{k} \left(8 \cdot n \cdot 2^{-10^{-6} \delta^2 d \ell}\right)^{k' \ell} \\
&\leq \sum_{k'=1}^{k} 0.01^{k' \ell}
\,\leq\, 0.95,
\end{align*}
where to obtain the second line we use that $N \le 1000 k n^{(1-\delta)\ell}$ and that $\ell \ge k \ge 20$, and in the final line we use that $d \ge 10^7 \log n / (\ell \delta)^2$.
This completes our proof.
\end{proof}

\section{Tight Running Time Lower Bounds}\label{sec:time}

In this section, we prove our tight running time lower bounds for $O(1)$-approximating Densest $k$-Subgraph, $O(1)$-approximating $k$-Biclique and (exact) Graph Pattern Detection.

\subsection{Constant Approximation for Densest $k$-Subgraph}

We start with the $n^{\Omega(k)}$ running time lower bound for $O(1)$-approximating Densest $k$-Subgraph, from which the remaining results easily follow. We remark that this running time lower bound improves upon that of $n^{\Omega(\log k)}$, which is implicit in~\cite{ChalermsookCKLM17}.

\begin{theorem}
\label{thm:dks-time}
Assuming Hypothesis~\ref{hyp:strong-planted-clique}, for any constant $C > 0$, there is no $f(k) \cdot n^{o(k)}$-time algorithm that can approximate Densest $k$-Subgraph to within a factor $C$ for any function $f$. Furthermore, this holds even in the perfect completeness case where the input graph is promised to contain a $k$-clique.
\end{theorem}

We will prove Theorem~\ref{thm:dks-time} by simply selecting an appropriate setting of parameters for the randomized graph product. Specifically, we will let $\ell = O(C \log n / k)$ and $N = n^{O(\ell)} = n^{o(\log n)}$; the generic soundness lemma (Lemma~\ref{lem:soundness-generic}) then implies that the density of any $k$-subgraph in the soundness case is at most $O(k / C)$ which yields the desired $C$ inapproximability factor.

\begin{proof}[Proof of Theorem~\ref{thm:dks-time}]
We will reduce the problem of distinguishing samples from $\cG(n, \frac{1}{2})$ vs. $\cG(n,\frac{1}{2}, \lceil n^{\delta} \rceil)$ to approximating Densest $k$-subgraph.

For $C$ the constant specified in the statement of the theorem, choose $\ell = \lceil \frac{10^8 C \log n}{\delta^2 k} \rceil$ and $N = \lceil 100 k n^{(1-\delta)\ell} \rceil$, so that $d = \frac{10^7 \log n}{\ell \delta^2} \le \frac{k}{10 C}$.
Given a graph $G$, we sample $G' \sim \rgp_{N,\ell}(G)$.

\noindent {\bf Completeness:} If $G \sim \cG(n,\frac{1}{2}, \lceil n^{\delta} \rceil)$, then by Lemma~\ref{lem:completeness} and since $N \ge 10 k n^{(1-\delta)\ell}$, we have that with probability at least $1-\exp(-4k) \ge 0.9$, $G' = \rgp_{N,\ell}(G)$ has a $k$-clique, and so $\den_{\le k}(G') \ge k-1$.

\noindent{\bf Soundness: } For any $20 \le k \le \ell$ and any $\delta$ bounded away from $0$, we clearly satisfy the requirement of Lemma~\ref{lem:soundness-generic} that $k\ell \le n^{0.99 \delta}$ and that $N \le 1000 k n^{(1-\delta)\ell}$ for any sufficiently large $n$.
Hence, if $G \sim \cG(n, \frac{1}{2})$, applying the Lemma to $G' \sim \rgp_{N,\ell}(G)$ we have that with probability at least $0.9$, $\den_{\le k}(G') \le d \le \frac{k}{10 C}$.

Thus, any algorithm which approximates Densest $k$-subgraph up to a factor of $C$ in time $f(k) N^{\epsilon k}$ can solve the planted $\lceil n^{\delta} \rceil$-clique problem in time $f(k) (100 k n^{(1-\delta)\ell})^{\epsilon k} = g(k) n^{\frac{(1-\delta)\delta^2 C}{10^8} \epsilon \log n}$ for $g(k) = f(k)(100k)^{\eps k}$.
This contradicts Hypothesis~\ref{hyp:strong-planted-clique} whenever $\lim_{n \to \infty} \eps = 0$.
Choosing $k$ to be a sufficiently slowly growing function of $n$, for any $\epsilon$ decreasing in $k$ we have a contradiction of the Hypothesis.
This concludes the proof.
\end{proof}

\subsection{$O(1)$-Approximation for $k$-Biclique}
\label{subsec:biclique}

Recall that a $k$-biclique, denoted by $K_{k, k}$, is a complete bipartite graph where each side has exactly $k$ vertices. In the \emph{$k$-Biclique} problem, we are given an undirected graph $G$ and a positive integer $k$. Further, we are promised that $G$ contains a $K_{k, k}$ as a subgraph. Our goal is to output a balanced biclique in $G$ of size as large as possible. (Note that we say that an algorithm achieves $\alpha$-approximation if the output biclique has size at least $k / \alpha$.)

We prove the following tight running time lower bound for $O(1)$-approximation of $k$-Biclique:

\begin{corollary} \label{cor:biclique}
Assuming Hypothesis~\ref{hyp:strong-planted-clique}, for any constant $C > 0$, there is no $f(k) \cdot n^{o(k)}$-time algorithm that can approximate $k$-Biclique to within a factor $C$ for any function $f$. Furthermore, this holds even when we are promised that the input graph contains a $2k$-clique.
\end{corollary}

\begin{proof}[Proof of Corollary~\ref{cor:biclique}]
Suppose contrapositively that there is an $f(k) \cdot n^{o(k)}$-time algorithm $\bA$ that can approximate $k$-Biclique to within a factor of $C$ for some function $f$. We may use it to approximate Densest $k$-Subgraph with perfect completeness as follows. On a given instance $(G, k)$ of Densest $k$-Subgraph with perfect completeness, we run our algorithm $\bA$ on $(G, \lfloor k / 2 \rfloor)$. From the approximation guarantee of $\bA$, we will find a $t$-Biclique for $t \geq \lfloor k/2 \rfloor / C \geq \frac{k}{4C}$. By taking this biclique together with arbitrary $(k - 2t)$ additional vertices, we find a subset of size $k$ that induces at least
$t^2 \geq \frac{k^2}{16C^2}$ edges. Hence, we have found a $16C^2$-approximation for Densest $k$-Subgraph in time $f(k) \cdot n^{o(k)}$, which by Theorem~\ref{thm:dks-time} violates Hypothesis~\ref{hyp:strong-planted-clique}.
\end{proof}

\subsection{Pattern detection}
\label{subsec:pattern-detection}

Theorem~\ref{thm:dks-time} also yields the following corollary for the running time of pattern detection:
\begin{corollary}\label{cor:not-induced}
Assuming Hypothesis~\ref{hyp:strong-planted-clique}, for almost all $k$-node patterns $H$, the not necessarily induced pattern detection problem cannot be solved in time $f(k) \cdot n^{o(k)}$.
\end{corollary}
\noindent In the statement of Corollary~\ref{cor:not-induced}, by ``almost all $k$-node patterns'' we mean w.h.p. over $H \sim \cG(k,1/2)$.
\footnote{It may be more natural to sample uniformly from all {\em unlabeled} $k$-node patterns, but w.h.p.~a graph drawn from $\cG(k,1/2)$ has no non-trivial automorphisms~\cite{erdHos1963asymmetric}, so the two notions of ``almost all $k$-node patterns'' are in fact equivalent.}
\begin{proof}
By standard concentration inequalities, most $H \sim \cG(k,\frac{1}{2})$ have average degree $\frac{k}{2} \pm o(k)$. For such a pattern $H$, an algorithm that solves the not necessarily induced pattern detection problem also gives $O(1)$-approximation for Densest $k$-Subgraph. Hence, the corollary immediately follows from Theorem~\ref{thm:dks-time}.
\end{proof}

\restatetheorem{thm:induced-pattern}
\begin{proof}
We assume that $H$ is at least $\frac{k}{4}$-dense (that is, has average degree at least $\frac{k}{2}$). This is without loss of generality as otherwise, we may take the complement of $H$ --- for induced pattern detection the complexity is the same.

We start our reduction from an instance $G$ of Densest $k$-Subgraph as in Theorem~\ref{thm:dks-time}.
We randomly color all the vertices of $G$ in $k$ colors, one for each vertex of $H$. 
We construct a graph $G'$ from $G$ by keeping edge $(u,v)\in G$ if and only if $u,v$ are have different colors, and the vertices in $H$ corresponding to those colors share an edge. (Note that we do not add any edges.)
\begin{description}
\item[Completeness:] If $G$ has a $k$-clique, with probability at least $k^{-k}$ it has $k$ different colors. In this case, the same vertices in $G'$ will form an induced copy of $H$.
\item[Soundness:] If $G$ does not have a $\frac{k}{4}$-dense $k$-subgraph, this remains true after removing edges. Hence $G'$ also does not contain any $\frac{k}{4}$-dense $k$-subgraph --- and in particular no copy of $H$.
\end{description}
As a result, if we can solve the induced  pattern detection problem in time $f(k) \cdot n^{o(k)}$, we can achieve $4$-approximation for Densest $k$-Subgraph with probability $k^{-k}$ in time $f(k) \cdot n^{o(k)}$. By repeating this construction $100 \cdot k^k$ times, we can achieve $4$-approximation for Densest $k$-Subgraph with probability 0.99 in time $O(f(k) \cdot k^k) \cdot n^{o(k)}$. Together with Theorem~\ref{thm:dks-time}, this concludes our proof.
\end{proof}

\section{Tight Inapproximability Results}\label{sec:inapprox}

In this section, we prove tight polynomial-time inapproximability results for Densest $k$-Subgraph, Smallest $k$-Edge Subgraph, Steiner $k$-Forest, Directed Steiner Network, and Densest $k$-Subhypergraph.

\subsection{Densest $k$-Subgraph}

We start with the $o(k)$-factor hardness of Densest $k$-Subgraph, from which our other results follow. 

\restatetheorem{thm:dks-inapprox}

The proof of Theorem~\ref{thm:dks-inapprox} is once again via selecting an appropriate setting of parameters for the randomized graph product. Specifically, we will select $\ell = O((\log n) \cdot g(k) / k)$ where $g(k) = o(k)$ is the assumed approximation ratio and $N = n^{O(\ell)} = n^{o(\log n)}$; the generic soundness lemma (Lemma~\ref{lem:soundness-generic}) then implies that the density of any $k$-vertex subgraph in the soundness case is at most $k / g(k)$ as desired. The arguments are formalized below.

\begin{proof}[Proof of Theorem~\ref{thm:dks-inapprox}]
We will reduce the problem of distinguishing samples from $\cG(n, \frac{1}{2})$ vs. $\cG(n,\frac{1}{2}, \lceil n^{\delta} \rceil)$ to approximating Densest $k$-subgraph on an $N$ vertex graph.

Let $g(k) \le k$ be any function growing with $k$.
Choose $\ell = \lceil \frac{10^8 g(k) \log n}{\delta^2 k} \rceil$ and $N = \lceil 100 k n^{(1-\delta)\ell} \rceil$, so that $d = \frac{10^7 \log n}{\ell \delta^2} \le \frac{k}{10 g(k)}$.
Given a graph $G$, we sample $G' \sim \rgp_{N,\ell}(G)$.

\noindent {\bf Completeness:} If $G \sim \cG(n,\frac{1}{2}, \lceil n^{\delta} \rceil)$, 
then just as in the proof of Theorem~\ref{thm:dks-time}, Lemma~\ref{lem:completeness} implies that with probability at least $0.9$, $G' = \rgp_{N,\ell}(G)$ has a $k$-clique.

\noindent{\bf Soundness: } For any $20 \le k \le \ell$ and $\delta$ bounded away from $0$, we satisfy the requirements of Lemma~\ref{lem:soundness-generic} (just as in the proof of Theorem~\ref{thm:dks-time}).
Applying the lemma, if $G \sim \cG(n, \frac{1}{2})$, we conclude that with probability at least $0.9$, $\den_{\le k}(G') \le d \le \frac{k}{10 g(k)}$. This means that any $k$-vertex subgraph of $G'$ contains at most $\frac{k^2}{10g(k)} < \frac{\binom{k}{2}}{g(k)}$ edges.

Hence, any algorithm which approximates Densest $k$-subgraph within $g(k)$ in time $f(k) \poly(N)$ can solve the planted $\lceil n^{\delta} \rceil$-clique problem in time $f(k)\cdot \poly(100 k n^{(1-\delta)\ell}) = h(k)\cdot \poly\left(n^{\frac{(1-\delta)\delta^2}{10^8}\frac{ g(k)}{k} \log n}\right)$ for $h(k) = f(k)\poly(100k)$.
Choosing $k$ to be a sufficiently slowly growing function of $n$, for any function $g(k)$ with $\lim_{n \to 0} \frac{g(k)}{k} = 0$ we have a contradiction of Hypothesis~\ref{hyp:strong-planted-clique}, as desired.
\end{proof}

\subsection{Smallest $k$-Edge Subgraph}

Recall that, in the \emph{Smallest $k$-Edge Subgraph}, we are given an undirected graph $G$ and $k \in \N$. The goal is to output a set $S$ of vertices of $G$ such that the induced subgraph on $S$ contains at least $k$ edges. The \emph{perfect completeness case} of Smallest $k$-Edge Subgraph is a special case where $k = \binom{\kappa}{2}$ for some $\kappa \in \N$ and it is promised that $G$ contains a clique of size $\kappa$.

\begin{corollary} \label{cor:smallest-subgraph}
Assuming Hypothesis~\ref{hyp:strong-planted-clique}, there is no $f(k) \cdot \poly(n)$-time algorithm that can approximate Smallest $k$-Edge Subgraph to within a factor $o(\sqrt{k})$ for any function $f$. Furthermore, this holds even in the perfect completeness case.
\end{corollary}

We prove the above corollary by reducing from Densest $k$-Subgraph;
we remark here that similar reductions between the two problems are folklore and have appeared before in literature, e.g. in~\cite{HajiaghayiJ06}. However, the exact statements proved before were slightly different, so we provide the full proof here for completeness.

\begin{proof}[Proof of Corollary~\ref{cor:smallest-subgraph}]
Suppose contrapositively that there is an $f(p) \cdot \poly(n)$-time algorithm $\bA$ that can approximate Smallest $p$-Edge Subgraph in the perfect completeness case to within a factor of $\sqrt{p} / g(p)$ for some $g(p) = \omega(1)$ and $f(p)$. We will construct an algorithm $\bB$ that achieves $o(k)$-approximation for Densest $k$-Subgraph in the perfect completeness case in time $h(k) \cdot \poly(n)$ for some function $h$. Our corollary then follows from Theorem~\ref{thm:dks-inapprox}.

Given an instance $(G = (V, E), k)$ of Densest $k$-Subgraph, Algorithm $\bB$ works as follows:
\begin{itemize}
\item Run algorithm $\bA$ on $G$ with $p = \binom{k}{2}$. Let the output of $\bA$ be $S$.
\item Enumerate to find a subset $T^* \subseteq S$ of size $k$ that maximizes $E[T^*]$.
\item Output $T^*$.
\end{itemize}
From the assumed approximation guarantee of $\bA$, we have $|S| \leq k \cdot \sqrt{p} / g(p) \leq k^2 / g(p)$. As a result, the running time of the second step is at most $\binom{|S|}{k} \poly(n) \leq 2^{O(k \log k)}\poly(n)$. Hence, the running time of the entire algorithm $\bB$ is $(2^{O(k \log k)} + f(O(k^2))) \cdot \poly(n)$ as desired. 

As for the approximation guarantee, let $T$ be a random subset of $S$ of size $k$. Notice that each edge in $E[S]$ belongs to $T$ with probability $\frac{k(k - 1)}{|S|(|S| - 1)}$. As a result, we have
\begin{align*}
E[T^*] 
\geq \E_T[E[T]] 
= \frac{k(k - 1)}{|S|(|S| - 1)} \cdot E[S] 
\geq \frac{k(k - 1)}{|S|(|S| - 1)} \cdot p 
\geq \frac{k^2 / 2}{(k^2 / g(p))^2} \cdot (k^2 / 4)
= (g(p))^2 / 8.
\end{align*}
Since $g(p) \to \infty$ as $p \to \infty$, the approximation ratio of $\bB$ is $o(k)$ as desired.
\end{proof}

\subsection{Steiner $k$-Forest}

Recall that, in the \emph{Steiner $k$-Forest} problem (aka the \emph{$k$-Forest} problem), we are given an edge-weighted undirected graph $G = (V, E)$, a set $\{(s_1, t_1), \dots, (s_\ell, t_\ell)\}$ of demand pairs and $k \in \N$. The goal is to find a (not necessarily induced) subgraph of $G$ with smallest total edge weight that connects at least $k$ demand pairs.

\begin{corollary} \label{cor:k-forest}
Assuming Hypothesis~\ref{hyp:strong-planted-clique}, there is no $f(k) \cdot \poly(n)$-time algorithm that can approximate Steiner $k$-Forest to within a factor $o(\sqrt{k})$ for any function $f$.
\end{corollary}

Corollary~\ref{cor:k-forest} immediately follows from Corollary~\ref{cor:smallest-subgraph} via the following reduction presented in \cite{HajiaghayiJ06}:

\begin{theorem}[{\cite[Theorem 6.2]{HajiaghayiJ06}}]
If there is an $f(k) \cdot \poly(n)$-time $g(k)$-approximation algorithm for Steiner $k$-Forest, then there is an $f(k) \cdot \poly(n)$-time $g(k)$-approximation algorithm for Smallest $k$-Edge Subgraph.
\end{theorem}

For interested readers, we remark that this reduction is as follows: for a given instance $(G = (V, E), k)$ of Smallest $k$-Edge Subgraph, create a star graph (with uniform edge weights) on vertices $\{c\} \cup V$ where $c$ is the new vertex which is the center of the star. Then, the demand pairs are exactly $E$. It is clear that there is a one-to-one correspondence between a solution of Steiner $k$-Forest problem and a solution of Small $k$-Edge Subgraph of the same cost.

\subsection{Directed Steiner Network}

Recall that, in the \emph{Directed Steiner Network} problem (aka the \emph{Directed Steiner Forest} problem), we are given an arc-weighted directed graph $G = (V, E)$, a set $\{(s_1, t_1), \dots, (s_k, t_k)\}$ of $k$ demand pairs. The goal is to find a (not necessarily induced) subgraph of $G$ with smallest total arc weight in which $t_i$ is reachable from $s_i$ for all $i \in [k]$.

\begin{corollary} \label{cor:dsn}
Assuming Hypothesis~\ref{hyp:strong-planted-clique}, there is no $f(k) \cdot \poly(n)$-time algorithm that can approximate Directed Steiner Network with $k$ terminal pairs to within a factor $o(\sqrt{k})$ for any function $f$.
\end{corollary}

We now prove Corollary~\ref{cor:dsn} via a reduction from Smallest $p$-Edge Subgraph. This reduction is similar to the reduction from the Label Cover problem by Dodis and Khanna~\cite{DodisK99}, and was also used in subsequent works~\cite{ChitnisFM18,DinurM18}. Since we start from a slightly different problem, we repeat the argument below for completeness.

\begin{proof}[Proof of Corollary~\ref{cor:dsn}]
Suppose contrapositively that there is an $f(k) \cdot \poly(n)$-time algorithm $\bA$ that can approximate Directed Steiner Network to within a factor of $o(\sqrt{k})$ for some function $f$. We will construct an algorithm $\bB$ that achieves $o(\sqrt{p})$-approximation for Smallest $p$-Subgraph in the perfect completeness case in time $h(p) \cdot \poly(n)$ for some function $h$. Our corollary then follows from Corollary~\ref{cor:smallest-subgraph}. 

Below, we will present an algorithm $\bB$ that succeeds with probability only $k^{-\Omega(k)}$. This suffices for us, since we may repeat this algorithm $k^{O(k)}$ times and output the best solution found (as in the proof of Theorem~\ref{thm:induced-pattern}). Given an instance $(G = (V, E), p)$ of Smallest $p$-Edge Subgraph where $p = \binom{k}{2}$, Algorithm $\bB$ works as follows:
\begin{itemize}
\item Let $V_1 \cup \cdots \cup V_k$ be a uniformly random partition of the vertex set $V$.
\item Construct an edge-weighted directed graph $G' = (V', E')$ as follows:
\begin{itemize}
\item For every $V_i$, create two copies of it: $V_i^1 = \{v^1 \mid v \in V_i\}$ and $V_i^2 = \{v^2 \mid v \in V_i\}$. \\ For brevity, let $V^1 = V_1^1 \cup \cdots \cup V_k^1$ and $V^2 = V_1^2 \cup \cdots \cup V_k^2$.
\item Let $V'$ be $V^1 \cup V^2 \cup \{s_1, \dots, s_k\} \cup \{t_1, \dots, t_k\}$ where $s_1, \dots, s_k, t_1, \dots, t_k$ are new vertices.
\item For every $v^1 \in V^1$ and $u^2 \in V^2$, add $(v^1, u^2)$ of weight 0 to $E'$ if and only if $\{v, u\} \in E$ or $v = u$.
\item For every $i \in [k]$ and $v \in V^1_i$, add an arc $(s_i, v^1)$ of weight 1 to $E'$
\item For every $j \in [k]$ and $v \in V^2_j$, add an arc $(v^2, t_j)$ of weight 1 to $E'$.
\item Let the demand pairs be $\{(s_i, t_j)\}_{i \in [k], j \in [k]}$.
\end{itemize}
\item Run algorithm $\bA$ on $(G', \{(s_i, t_j)\}_{i \in [k], j \in [k]})$. Let the output of $\bA$ be $R$.
\item Let $S$ be the set of vertices $v \in V$ such that $(s_i, v^1)$ or $(v^2, t_i)$ belongs to $R$ for some $i \in [k]$.
\item Output $S$.
\end{itemize}
Suppose that there exists a $k$-clique $C$ in $G$. With probability at least $k^{-k}$, each vertex of $C$ belongs to a different partition $V_i$. When this is the case, we may take the all 0-weighted arcs, and the 1-weighted arcs adjacent to $C$ as a solution. This results in a solution of cost $2k$. From the assumed approximation guarantee of $\bA$, we must have $|S| \leq 2k \cdot o(\sqrt{k^2}) = k \cdot o(\sqrt{p})$. It is also simple to see that $S$ must induce at least one edge in $G$ between $V_i, V_j$ for every $i,j \in [k]$. This means that $S$ is a solution of Smallest $p$-Edge Subgraph of size at most $k \cdot o(\sqrt{p})$. Hence, this is an $o(\sqrt{p})$-approximate solution as desired. Finally, it is obvious to see that $\bB$ runs in time $f(O(p)) \cdot \poly(n)$.
\end{proof}

\subsection{Densest $k$-Subhypergraph}

In the \emph{Densest $k$-Subhypergraph} problem, we are given a hypergraph $G = (V, E)$ and an integer $k$. The goal is to output a set $S \subseteq V$ of $k$ vertices that maximizes the number of hyperedges fully contained in $S$. We remark here that each hyperedge in this hypergraph may have an unbounded size.

\begin{theorem} \label{thm:dksh}
Assuming Hypothesis~\ref{hyp:strong-planted-clique}, there is no $f(k) \cdot \poly(n)$-time algorithm that can approximate Densest $k$-Subhypergraph to within a factor $2^{o(k)}$ for any function $f$.
\end{theorem}

The proof of our inapproximability for Densest $k$-Subhypergraph (Theorem~\ref{thm:dksh}) is unlike the others in this section: instead of starting from the inapproximability of Densest $k$-Subgraph (Theorem~\ref{thm:dks-inapprox}), we will start from the tight running time lower bound for $O(1)$-approximate $k$-Biclique (Corollary~\ref{cor:biclique}). 

\subsubsection{A Combinatorial Lemma}

Before we describe our reduction, let us state and prove the following lemma, which bounds the number of (induced) copies of $K_{\ell, \ell}$ in a $K_{t, t}$-free graph for $\ell < t$. This is a generalized setting of the classic Kovari-Sos-Turan theorem~\cite{kHovari1954problem}, which applies only to the case $\ell = 1$. Note that, due to the parameters of interest in our reduction, we only prove a good bound for large $\ell$; for $\ell = 1$, the bound in the lemma below is trivial (and hence weaker than the Kovari-Sos-Turan theorem).

\begin{lemma} \label{lem:num-biclique}
Let $\kappa, t, \ell$ be positive integers such that $\ell < t \leq \kappa / 16$. Then, for any $\kappa$-vertex $K_{t, t}$-free graph $H$, the number of (not necessarily induced) copies of $K_{\ell, \ell}$ in $H$ is at most 
\begin{align*}
\left(2 \cdot e^{-\frac{\ell^2}{16t}}\right) \cdot \binom{\kappa}{\ell}\binom{\kappa - \ell}{\ell}.
\end{align*}
\end{lemma}

\begin{proof}
Let $V$ denote the vertex set of $H$. We will count the number of copies of $K_{\ell, t}$ in $H$ in two ways. First, for every subset $S \in \binom{V}{\ell}$, the number of $(\ell, t)$-bicliques of the form $(S, T)$ where $T \in \binom{V}{t}$ is $\binom{|N(S)|}{t}$ where $N(S)$ denote the set of common neighbors of $S$. Hence, in total the number of $(\ell, t)$-bicliques in $H$ is $$\sum_{S \in \binom{V}{\ell}} \binom{|N(S)|}{t}.$$
On the other hand, for every set $T \in \binom{V}{t}$, we must have $|N(T)| \leq t - 1$ since $H$ does not contain $K_{t, t}$. As a result, the number of $(\ell, t)$-bicliques of the form $(S, T)$ for a fixed $T$ is at most $\binom{t - 1}{\ell}$. That is, in total, there can be at most $\binom{\kappa}{t} \binom{t - 1}{\ell}$ copies of $K_{\ell, t}$ in the graph. This implies that
\begin{align} \label{eq:main}
\binom{\kappa}{t} \binom{t - 1}{\ell} \geq \sum_{S \in \binom{V}{\ell}} \binom{|N(S)|}{t}.
\end{align}

For brevity, let $\lambda$ denote $(\kappa - \ell)/t$, $\gamma$ denote $\lambda^{-\ell/(2t)}$ and let $x$ denote $\gamma \cdot \kappa + (1 - \gamma) \cdot t$ (note $x$ is chosen so that $\frac{x - t}{\kappa - t} = \gamma$). Let us now classify $S \in \binom{A}{\ell}$ into two groups: those with $|N(S)| \geq x$ and those with $|N(S)| < x$. That is, we define
\begin{align*}
\calS_{\geq x} &:= \left\{S \in \binom{V}{\ell} \middle\vert |N(S)| \geq x\right\}, \qquad\text{and}\qquad
\calS_{< x} := \left\{S \in \binom{V}{\ell} \middle\vert |N(S)| < x\right\}.
\end{align*}

From~\eqref{eq:main}, we have
\begin{align*}
\binom{\kappa}{t}\binom{t-1}{\ell} \geq \sum_{S \in \calS_{\geq x}} \binom{|N(S)|}{t} \geq |\calS_{\geq x}| \cdot \binom{\lceil x \rceil}{t}.
\end{align*}
Rearranging, we have
\begin{align}
|\calS_{\geq x}| 
\leq \frac{\binom{t - 1}{\ell} \binom{\kappa}{t}}{ \binom{\lceil x \rceil}{t}} 
\leq \binom{t - 1}{\ell} \left(\frac{\kappa - t}{x - t}\right)^t
= \binom{t - 1}{\ell} \gamma^{-t}
&\leq \binom{\kappa}{\ell} \left(\frac{t - 1}{\kappa}\right)^{\ell} \gamma^{-t}\nonumber\\
&\leq \binom{\kappa}{\ell} \lambda^{-\ell} \gamma^{-t}
= \binom{\kappa}{\ell} \lambda^{-\ell/2}
 \leq \binom{\kappa}{\ell}2^{-\ell/2}, 
\label{eq:tail-bound}
\end{align}
\noindent where the last line follows because  $t, \ell \leq \frac{1}{3}\kappa$.

Let us now count the number of $K_{\ell, \ell}$ in $H$. For each fixed $S \in \binom{V}{\ell}$, the number of $K_{\ell, \ell}$ of the form $(S, T)$ where $T \in \binom{V}{\ell}$ is exactly $\binom{|N(S)|}{\ell}$. Thus, the total number of $K_{\ell, \ell}$ in $G$ is
\begin{align*}
\sum_{S \in \binom{V}{\ell}} \binom{|N(S)|}{\ell}.
\end{align*}

This term can be further bounded by
\begin{align}
\sum_{S \in \binom{V}{\ell}} \binom{|N(S)|}{\ell} 
&= \sum_{S \in \calS_{\geq x}} \binom{|N(S)|}{\ell} + \sum_{S \in \calS_{< x}} \binom{|N(S)|}{\ell} \nonumber \\
&\leq |\calS_{\geq x}| \binom{\kappa - \ell}{\ell} + |\calS_{< x}| \binom{x}{\ell} \nonumber \\
&\overset{\eqref{eq:tail-bound}}{\leq} 2^{-\ell/2}\binom{\kappa}{\ell}\binom{\kappa - \ell}{\ell} + \binom{\kappa}{\ell}\binom{\lfloor x \rfloor}{\ell} \nonumber \\
&\leq 2^{-\ell/2}\binom{\kappa}{\ell}\binom{\kappa - \ell}{\ell} + \binom{\kappa}{\ell}\binom{\kappa - \ell}{\ell} \cdot \left(\frac{x}{\kappa - \ell}\right)^\ell \nonumber \\
&\leq 2^{-\ell/2}\binom{\kappa}{\ell}\binom{\kappa - \ell}{\ell} + \binom{\kappa}{\ell}\binom{\kappa - \ell}{\ell} \cdot \left(1 - \frac{(1 - \gamma)(\kappa - t) - \ell}{\kappa - \ell}\right)^\ell \nonumber \\
&= 2^{-\ell/2}\binom{\kappa}{\ell}\binom{\kappa - \ell}{\ell} + \binom{\kappa}{\ell}\binom{\kappa - \ell}{\ell} \cdot \left(1 - \frac{(1 - \gamma - \frac{\ell}{\kappa - t})(\kappa - t)}{\kappa - \ell}\right)^\ell \nonumber \\
(\text{From } \ell < t \leq \kappa/2) &\leq 2^{-\ell/2}\binom{\kappa}{\ell}\binom{\kappa - \ell}{\ell} + \binom{\kappa}{\ell}\binom{\kappa - \ell}{\ell} \cdot \left(1 - 0.5(1 - \gamma - 2\ell/\kappa)\right)^\ell \nonumber \\
&\leq 2^{-\ell/2}\binom{\kappa}{\ell}\binom{\kappa - \ell}{\ell} + \binom{\kappa}{\ell}\binom{\kappa - \ell}{\ell} \cdot e^{-0.5\ell(1 - \gamma - 2\ell/\kappa)} \label{eq:bound2}
\end{align}

Consider the term $(1 - \gamma - 2\ell/\kappa)$. We can bound it as follows:
\begin{align*}
(1 - \gamma - 2\ell/\kappa) 
&= \left(1 - \frac{1}{\lambda^{\ell/2t}} - \frac{2\ell}{\kappa}\right) \\
(\text{From } \lambda = (\kappa - \ell) / t \geq 2) &\geq \left(1 - 0.5^{\ell/2t} - \frac{2\ell}{\kappa}\right) \\
(\text{Bernoulli inequality}) &\geq \left(1 - \left(1 - \frac{\ell}{4t}\right) - \frac{2\ell}{\kappa}\right) \\
&= \left(\frac{\ell}{4t} - \frac{2\ell}{\kappa}\right) \\
(\text{From } \kappa \geq 16t) &\geq \frac{\ell}{8t}.
\end{align*}
Plugging this back into~\eqref{eq:bound2}, we have
\begin{align*}
\sum_{S \in \binom{A}{\ell}} \binom{|N(S)|}{\ell} 
&\leq \binom{\kappa}{\ell}\binom{\kappa - \ell}{\ell} \cdot \left(2^{-\ell/2} + e^{-\frac{\ell^2}{16t}}\right) \\
(\text{From } \ell < t) &\leq \binom{\kappa}{\ell}\binom{\kappa - \ell}{\ell} \cdot \left(2 \cdot e^{-\frac{\ell^2}{16t}}\right).
\end{align*}
This concludes our proof.
\end{proof}

\subsubsection{Proof of Theorem~\ref{thm:dksh}}

We are now ready to prove Theorem~\ref{thm:dksh}.

\begin{proof}[Proof of Theorem~\ref{thm:dksh}]
Suppose contrapositively that there is an $f(\rho) \cdot \poly(n)$-time $2^{\rho/g(\rho)}$-approximation algorithm $\bA$ for Densest $\rho$-Subhypergraph where $g = \omega(1)$. We will construct an algorithm $\bB$ that achieves $O(1)$-approximation for $k$-Biclique with the promise that a $2k$-Clique exists in time $h(k) \cdot n^{o(k)}$ for some function $h$. Theorem~\ref{thm:dksh} then follows from Corollary~\ref{cor:biclique}.

We define $\bB$ on input $(G = (V, E), k)$ as follows:
\begin{itemize}
\item Let $\rho = 2k$ and $\ell = \lceil \rho/g(\rho)^{0.1} \rceil$.
\item Construct a hypergraph $G'$ with the same vertex set as $G$, and we add a hyperedge $e = \{u_1, \dots, u_{2\ell}\}$ to $G'$ for all distinct $u_1, \dots, u_{2\ell} \in V$ that induce a $2\ell$-clique in $G$. 
\item Run $\bA$ on $(G', \rho)$. Let $S$ be $\bA$'s output.
\item Use brute force to find a maximum balanced biclique in $S$ and output it.
\end{itemize}
Notice that algorithm $\bB$ runs in time $(f(2k) + 2^{O(k)}) \cdot n^{O(\ell)} = (f(2k) + 2^{O(k)}) \cdot n^{o(k)}$ as desired, where the term $2^{O(k)}$ comes from the last step.

Next, we claim that, when $k$ is sufficiently large, the output biclique has size at least $t := \lfloor k / 8 \rfloor$, which would give us the desired $O(1)$-approximation ratio. Suppose for the sake of contradiction that this is not true, i.e. that the induced graph $G[S]$ is $K_{t, t}$-free. 

For any sufficiently large $k$, we have $\ell < t$. Hence, we may apply Lemma~\ref{lem:num-biclique}, which gives the following upper bound on the number of not-necessarily induced copies of $K_{\ell, \ell}$ in $G[S]$:
\begin{align*}
\left(2\cdot e^{-\frac{\ell^2}{16t}}\right) \cdot \binom{2k}{\ell}\binom{2k - \ell}{\ell} \leq e^{-\Omega(k/g(k)^{0.2})} \cdot \binom{2k}{\ell}\binom{2k - \ell}{\ell}.
\end{align*}
However, since each hyperedge $e = \{u_1, \dots, u_{2\ell}\}$ corresponds to $\binom{2\ell}{\ell}$ copies of $K_{\ell, \ell}$, the number of hyperedges fully contained in $S$ is thus at most
\begin{align*}
e^{-\Omega(k/g(k)^{0.2})} \cdot \binom{2k}{\ell}\binom{2k - \ell}{\ell} / \binom{2\ell}{\ell} = e^{-\Omega(k/g(k)^{0.2})} \cdot \binom{2k}{2\ell},
\end{align*}
which is less than $2^{-\rho/g(\rho)} \cdot \binom{2k}{2\ell}$ for any sufficiently large $k$. This contradicts the approximation guarantee of $\bA$ since the optimal solution (corresponding to the $2k$-clique) contains $\binom{2k}{2\ell}$ hyperedges.
\end{proof}

\subsection*{Acknowledgment}
Pasin Manurangsi would like to thank Michal Pilipczuk and Daniel Lokshtanov for posing the approximability of Densest $k$-Subhypergraph as an open problem at Dagstuhl Seminar on New Horizons in Parameterized Complexity, and for helpful discussions.

\bibliographystyle{alpha}
\bibliography{refs}

\newcommand{\etalchar}[1]{$^{#1}$}
\begin{thebibliography}{MSOI{\etalchar{+}}02}

\bibitem[AAM{\etalchar{+}}11]{alon2011inapproximability}
Noga Alon, Sanjeev Arora, Rajsekar Manokaran, Dana Moshkovitz, and Omri
  Weinstein.
\newblock Inapproximability of densest $\kappa$-subgraph from average case
  hardness.
\newblock 2011.

\bibitem[AB18]{AbboudB18}
Amir Abboud and Greg Bodwin.
\newblock Reachability preservers: New extremal bounds and approximation
  algorithms.
\newblock In {\em SODA}, pages 1865--1883, 2018.

\bibitem[ABBG10]{ABBG10}
Sanjeev Arora, Boaz Barak, Markus Brunnermeier, and Rong Ge.
\newblock Computational complexity and information asymmetry in financial
  products (extended abstract).
\newblock In {\em ICS}, pages 49--65, 2010.

\bibitem[ABW10]{applebaum2010public}
Benny Applebaum, Boaz Barak, and Avi Wigderson.
\newblock Public-key cryptosystem from different assumptions.
\newblock In {\em STOC}, pages 171--180, 2010.

\bibitem[AFWZ95]{AlonFWZ95}
Noga Alon, Uriel Feige, Avi Wigderson, and David Zuckerman.
\newblock Derandomized graph products.
\newblock {\em Comput. Complex.}, 5(1):60--75, 1995.

\bibitem[Ale11]{alekhnovich2011more}
Michael Alekhnovich.
\newblock More on average case vs approximation complexity.
\newblock {\em Computational Complexity}, 20(4):755--786, 2011.

\bibitem[App13]{Applebaum13}
Benny Applebaum.
\newblock Pseudorandom generators with long stretch and low locality from
  random local one-way functions.
\newblock {\em {SIAM} J. Comput.}, 42(5):2008--2037, 2013.

\bibitem[AYZ97]{AYZ97}
Noga Alon, Raphael Yuster, and Uri Zwick.
\newblock Finding and counting given length cycles.
\newblock {\em Algorithmica}, 17(3):209--223, 1997.

\bibitem[Bar18]{Barman18}
Siddharth Barman.
\newblock Approximating {Nash} equilibria and dense subgraphs via an
  approximate version of {Carath{\'{e}}odory}'s theorem.
\newblock {\em {SIAM} J. Comput.}, 47(3):960--981, 2018.

\bibitem[BB19]{BB19}
Matthew Brennan and Guy Bresler.
\newblock Optimal average-case reductions to sparse {PCA:} from weak
  assumptions to strong hardness.
\newblock In {\em COLT}, pages 469--470, 2019.

\bibitem[BBB{\etalchar{+}}13]{BBBCT13}
Maria{-}Florina Balcan, Christian Borgs, Mark Braverman, Jennifer~T. Chayes,
  and Shang{-}Hua Teng.
\newblock Finding endogenously formed communities.
\newblock In {\em SODA}, pages 767--783, 2013.

\bibitem[BBH18]{BBH18}
Matthew Brennan, Guy Bresler, and Wasim Huleihel.
\newblock Reducibility and computational lower bounds for problems with planted
  sparse structure.
\newblock In {\em COLT}, pages 48--166, 2018.

\bibitem[BBH{\etalchar{+}}20]{brennan2020statistical}
Matthew Brennan, Guy Bresler, Samuel~B Hopkins, Jerry Li, and Tselil Schramm.
\newblock Statistical query algorithms and low-degree tests are almost
  equivalent.
\newblock {\em arXiv preprint arXiv:2009.06107}, 2020.

\bibitem[BBM{\etalchar{+}}13]{BBMRY13}
Piotr Berman, Arnab Bhattacharyya, Konstantin Makarychev, Sofya Raskhodnikova,
  and Grigory Yaroslavtsev.
\newblock Approximation algorithms for spanner problems and directed steiner
  forest.
\newblock {\em Inf. Comput.}, 222:93--107, 2013.

\bibitem[BCC{\etalchar{+}}10]{BCCFV10}
Aditya Bhaskara, Moses Charikar, Eden Chlamtac, Uriel Feige, and Aravindan
  Vijayaraghavan.
\newblock Detecting high log-densities: an {$O(n^{1/4})$} approximation for
  densest $k$-subgraph.
\newblock In {\em STOC}, pages 201--210, 2010.

\bibitem[BCKS16]{BCKS16}
Umang Bhaskar, Yu~Cheng, Young~Kun Ko, and Chaitanya Swamy.
\newblock Hardness results for signaling in bayesian zero-sum and network
  routing games.
\newblock In {\em EC}, pages 479--496, 2016.

\bibitem[BCV{\etalchar{+}}12]{BhaskaraCVGZ12}
Aditya Bhaskara, Moses Charikar, Aravindan Vijayaraghavan, Venkatesan
  Guruswami, and Yuan Zhou.
\newblock Polynomial integrality gaps for strong {SDP} relaxations of densest
  $k$-subgraph.
\newblock In {\em SODA}, pages 388--405, 2012.

\bibitem[BGH{\etalchar{+}}16]{BhangaleGHKK16}
Amey Bhangale, Rajiv Gandhi, Mohammad~Taghi Hajiaghayi, Rohit Khandekar, and
  Guy Kortsarz.
\newblock Bicovering: Covering edges with two small subsets of vertices.
\newblock In {\em ICALP}, pages 6:1--6:12, 2016.

\bibitem[BGKM18]{BGS18}
Arnab Bhattacharyya, Suprovat Ghoshal, {Karthik {C. S.}}, and Pasin Manurangsi.
\newblock Parameterized intractability of even set and shortest vector problem
  from {Gap-ETH}.
\newblock In {\em ICALP}, pages 17:1--17:15, 2018.

\bibitem[BGLR93]{BellareGLR93}
Mihir Bellare, Shafi Goldwasser, Carsten Lund, and A.~Russeli.
\newblock Efficient probabilistically checkable proofs and applications to
  approximations.
\newblock In {\em STOC}, pages 294--304, 1993.

\bibitem[BHK{\etalchar{+}}19]{BarakHKKMP19}
Boaz Barak, Samuel~B. Hopkins, Jonathan~A. Kelner, Pravesh~K. Kothari, Ankur
  Moitra, and Aaron Potechin.
\newblock A nearly tight sum-of-squares lower bound for the planted clique
  problem.
\newblock {\em {SIAM} J. Comput.}, 48(2):687--735, 2019.

\bibitem[BKRW17]{BravermanKRW17}
Mark Braverman, Young Kun{-}Ko, Aviad Rubinstein, and Omri Weinstein.
\newblock {ETH} hardness for densest-\emph{k}-subgraph with perfect
  completeness.
\newblock In {\em SODA}, pages 1326--1341, 2017.

\bibitem[BKS13]{barak2013optimality}
Boaz Barak, Guy Kindler, and David Steurer.
\newblock On the optimality of semidefinite relaxations for average-case and
  generalized constraint satisfaction.
\newblock In {\em ITCS}, pages 197--214, 2013.

\bibitem[BKS18]{BKS18}
Markus Bl{\"{a}}ser, Balagopal Komarath, and Karteek Sreenivasaiah.
\newblock Graph pattern polynomials.
\newblock In {\em FSTTCS}, pages 18:1--18:13, 2018.

\bibitem[BKW15]{BKW15}
Mark Braverman, Young Kun{-}Ko, and Omri Weinstein.
\newblock Approximating the best {Nash} equilibrium in $n^{o(\log n)}$-time
  breaks the exponential time hypothesis.
\newblock In {\em SODA}, pages 970--982, 2015.

\bibitem[BPR16]{BPR16}
Yakov Babichenko, Christos~H. Papadimitriou, and Aviad Rubinstein.
\newblock Can almost everybody be almost happy?
\newblock In {\em ITCS}, pages 1--9, 2016.

\bibitem[BR13]{BR13}
Quentin Berthet and Philippe Rigollet.
\newblock Complexity theoretic lower bounds for sparse principal component
  detection.
\newblock In {\em COLT}, pages 1046--1066, 2013.

\bibitem[BS89]{BermanS89}
Piotr Berman and Georg Schnitger.
\newblock On the complexity of approximating the independent set problem.
\newblock In {\em STACS}, pages 256--268, 1989.

\bibitem[BS16]{WBS16}
Tengyao Wang~Quentin Berthet and Richard~J Samworth.
\newblock Statistical and computational trade-offs in estimation of sparse
  principal components.
\newblock {\em The Annals of Statistics}, 2016.

\bibitem[CCC{\etalchar{+}}99]{CCCDGGL99}
Moses Charikar, Chandra Chekuri, To{-}Yat Cheung, Zuo Dai, Ashish Goel, Sudipto
  Guha, and Ming Li.
\newblock Approximation algorithms for directed steiner problems.
\newblock {\em J. Algorithms}, 33(1):73--91, 1999.

\bibitem[CCK{\etalchar{+}}17]{ChalermsookCKLM17}
Parinya Chalermsook, Marek Cygan, Guy Kortsarz, Bundit Laekhanukit, Pasin
  Manurangsi, Danupon Nanongkai, and Luca Trevisan.
\newblock From {Gap-ETH} to {FPT-Inapproximability}: Clique, dominating set,
  and more.
\newblock In {\em FOCS}, pages 743--754, 2017.

\bibitem[CDK12]{ChlamtacDK12}
Eden Chlamtac, Michael Dinitz, and Robert Krauthgamer.
\newblock Everywhere-sparse spanners via dense subgraphs.
\newblock In {\em FOCS}, pages 758--767, 2012.

\bibitem[CDK{\etalchar{+}}18]{ChlamtacDKKR18}
Eden Chlamt{\'{a}}c, Michael Dinitz, Christian Konrad, Guy Kortsarz, and George
  Rabanca.
\newblock The densest k-subhypergraph problem.
\newblock {\em {SIAM} J. Discret. Math.}, 32(2):1458--1477, 2018.

\bibitem[CDKL17]{CDKL17}
Eden Chlamt{\'{a}}c, Michael Dinitz, Guy Kortsarz, and Bundit Laekhanukit.
\newblock Approximating spanners and directed steiner forest: Upper and lower
  bounds.
\newblock In {\em SODA}, pages 534--553, 2017.

\bibitem[CDM17a]{ChlamtacDM17}
Eden Chlamt{\'{a}}c, Michael Dinitz, and Yury Makarychev.
\newblock Minimizing the union: Tight approximations for small set bipartite
  vertex expansion.
\newblock In {\em SODA}, pages 881--899, 2017.

\bibitem[CDM17b]{CDM17}
Radu Curticapean, Holger Dell, and D{\'{a}}niel Marx.
\newblock Homomorphisms are a good basis for counting small subgraphs.
\newblock In {\em STOC}, pages 210--223, 2017.

\bibitem[CEGS11]{ChekuriEGS11}
Chandra Chekuri, Guy Even, Anupam Gupta, and Danny Segev.
\newblock Set connectivity problems in undirected graphs and the directed
  steiner network problem.
\newblock {\em {ACM} Trans. Algorithms}, 7(2):18:1--18:17, 2011.

\bibitem[CFM18]{ChitnisFM18}
Rajesh Chitnis, Andreas~Emil Feldmann, and Pasin Manurangsi.
\newblock Parameterized approximation algorithms for bidirected steiner network
  problems.
\newblock In {\em ESA}, pages 20:1--20:16, 2018.

\bibitem[Cha00]{Charikar00}
Moses Charikar.
\newblock Greedy approximation algorithms for finding dense components in a
  graph.
\newblock In {\em APPROX}, pages 84--95, 2000.

\bibitem[CHK11]{CharikarHK11}
Moses Charikar, MohammadTaghi Hajiaghayi, and Howard~J. Karloff.
\newblock Improved approximation algorithms for label cover problems.
\newblock {\em Algorithmica}, 61(1):190--206, 2011.

\bibitem[CJ03]{CJ03}
Liming Cai and David~W. Juedes.
\newblock On the existence of subexponential parameterized algorithms.
\newblock {\em J. Comput. Syst. Sci.}, 67(4):789--807, 2003.

\bibitem[CKL{\etalchar{+}}18]{cygan-dksh}
Marek Cygan, Pawel Komosa, Daniel Lokshtanov, Michal Pilipczuk, Marcin
  Pilipczuk, and Saket Saurabh.
\newblock Randomized contractions meet lean decompositions.
\newblock {\em CoRR}, abs/1810.06864, 2018.

\bibitem[CM14]{CM14}
Radu Curticapean and D{\'{a}}niel Marx.
\newblock Complexity of counting subgraphs: Only the boundedness of the
  vertex-cover number counts.
\newblock In {\em FOCS}, pages 130--139, 2014.

\bibitem[CM18]{ChlamtacM18}
Eden Chlamt{\'{a}}c and Pasin Manurangsi.
\newblock Sherali-adams integrality gaps matching the log-density threshold.
\newblock In {\em APPROX}, pages 10:1--10:19, 2018.

\bibitem[CMMV17]{ChlamtacMMV17}
Eden Chlamt{\'{a}}c, Pasin Manurangsi, Dana Moshkovitz, and Aravindan
  Vijayaraghavan.
\newblock Approximation algorithms for label cover and the log-density
  threshold.
\newblock In {\em SODA}, pages 900--919, 2017.

\bibitem[CPP16]{CPP16}
Marek Cygan, Marcin Pilipczuk, and Michal Pilipczuk.
\newblock Known algorithms for edge clique cover are probably optimal.
\newblock {\em {SIAM} J. Comput.}, 45(1):67--83, 2016.

\bibitem[Dan16]{Daniely16}
Amit Daniely.
\newblock Complexity theoretic limitations on learning halfspaces.
\newblock In {\em STOC}, pages 105--117, 2016.

\bibitem[DF95]{DowneyF95-w1}
Rodney~G. Downey and Michael~R. Fellows.
\newblock Fixed-parameter tractability and completeness {II:} on completeness
  for {W[1]}.
\newblock {\em Theor. Comput. Sci.}, 141(1{\&}2):109--131, 1995.

\bibitem[DF13]{DowneyF13}
Rodney~G. Downey and Michael~R. Fellows.
\newblock {\em Fundamentals of Parameterized Complexity}.
\newblock Springer, 2013.

\bibitem[DK99]{DodisK99}
Yevgeniy Dodis and Sanjeev Khanna.
\newblock Design networks with bounded pairwise distance.
\newblock In {\em STOC}, pages 750--759, 1999.

\bibitem[DKN17]{DinitzKN17}
Michael Dinitz, Guy Kortsarz, and Zeev Nutov.
\newblock Improved approximation algorithm for steiner \emph{k}-forest with
  nearly uniform weights.
\newblock {\em {ACM} Trans. Algorithms}, 13(3):40:1--40:16, 2017.

\bibitem[DKS17]{DKS17}
S{\o}ren Dahlgaard, Mathias B{\ae}k~Tejs Knudsen, and Morten St{\"{o}}ckel.
\newblock Finding even cycles faster via capped $k$-walks.
\newblock In {\em STOC}, pages 112--120, 2017.

\bibitem[DM18]{DinurM18}
Irit Dinur and Pasin Manurangsi.
\newblock {ETH}-hardness of approximating 2-{CSP}s and directed steiner
  network.
\newblock In {\em ITCS}, pages 36:1--36:20, 2018.

\bibitem[DS16]{DanielyS16}
Amit Daniely and Shai Shalev{-}Shwartz.
\newblock Complexity theoretic limitations on learning {DNF}'s.
\newblock In {\em COLT}, pages 815--830, 2016.

\bibitem[DVW19]{DVVW19}
Mina Dalirrooyfard, Thuy~Duong Vuong, and Virginia~Vassilevska Williams.
\newblock Graph pattern detection: hardness for all induced patterns and faster
  non-induced cycles.
\newblock In {\em STOC}, pages 1167--1178, 2019.

\bibitem[ER63]{erdHos1963asymmetric}
Paul Erd{\H{o}}s and Alfr{\'e}d R{\'e}nyi.
\newblock Asymmetric graphs.
\newblock {\em Acta Mathematica Academiae Scientiarum Hungarica},
  14(3-4):295--315, 1963.

\bibitem[Fei02]{Feige02}
Uriel Feige.
\newblock Relations between average case complexity and approximation
  complexity.
\newblock In {\em STOC}, pages 534--543, 2002.

\bibitem[FGR{\etalchar{+}}17]{FGRVX17}
Vitaly Feldman, Elena Grigorescu, Lev Reyzin, Santosh~S. Vempala, and Ying
  Xiao.
\newblock Statistical algorithms and a lower bound for detecting planted
  cliques.
\newblock {\em J. {ACM}}, 64(2):8:1--8:37, 2017.

\bibitem[FK03]{FK03}
Uriel Feige and Robert Krauthgamer.
\newblock The probable value of the lov{\'{a}}sz--schrijver relaxations for
  maximum independent set.
\newblock {\em {SIAM} J. Comput.}, 32(2):345--370, 2003.

\bibitem[FKN12]{FeldmanKN12}
Moran Feldman, Guy Kortsarz, and Zeev Nutov.
\newblock Improved approximation algorithms for directed steiner forest.
\newblock {\em J. Comput. Syst. Sci.}, 78(1):279--292, 2012.

\bibitem[FKP01]{FPK01}
Uriel Feige, Guy Kortsarz, and David Peleg.
\newblock The dense \emph{k}-subgraph problem.
\newblock {\em Algorithmica}, 29(3):410--421, 2001.

\bibitem[FS97]{FS97}
Uriel Feige and Michael Seltser.
\newblock On the densest $k$-subgraph problem.
\newblock Technical report, Weizmann Institute of Science, Rehovot, Israel,
  1997.

\bibitem[GHNR10]{GuptaHNR10}
Anupam Gupta, Mohammad~Taghi Hajiaghayi, Viswanath Nagarajan, and R.~Ravi.
\newblock Dial a ride from \emph{k}-forest.
\newblock {\em {ACM} Trans. Algorithms}, 6(2):41:1--41:21, 2010.

\bibitem[GMZ17]{GMZ17}
Chao Gao, Zongming Ma, and Harrison~H Zhou.
\newblock Sparse {CCA}: adaptive estimation and computational barriers.
\newblock {\em The Annals of Statistics}, 2017.

\bibitem[GNS11]{GuoNS11}
Jiong Guo, Rolf Niedermeier, and Ondrej Such{\'{y}}.
\newblock Parameterized complexity of arc-weighted directed steiner problems.
\newblock {\em {SIAM} J. Discret. Math.}, 25(2):583--599, 2011.

\bibitem[HJ06]{HajiaghayiJ06}
Mohammad~Taghi Hajiaghayi and Kamal Jain.
\newblock The prize-collecting generalized steiner tree problem via a new
  approach of primal-dual schema.
\newblock In {\em SODA}, pages 631--640, 2006.

\bibitem[HK11]{HK11}
Elad Hazan and Robert Krauthgamer.
\newblock How hard is it to approximate the best {Nash} equilibrium?
\newblock {\em {SIAM} J. Comput.}, 40(1):79--91, 2011.

\bibitem[HWX15]{HWX15}
Bruce~E. Hajek, Yihong Wu, and Jiaming Xu.
\newblock Computational lower bounds for community detection on random graphs.
\newblock In {\em COLT}, pages 899--928, 2015.

\bibitem[IP01]{IP01}
Russell Impagliazzo and Ramamohan Paturi.
\newblock On the complexity of $k$-{SAT}.
\newblock {\em J. Comput. Syst. Sci.}, 62(2):367--375, 2001.

\bibitem[IPZ01]{IPZ01}
Russell Impagliazzo, Ramamohan Paturi, and Francis Zane.
\newblock Which problems have strongly exponential complexity?
\newblock {\em J. Comput. Syst. Sci.}, 63(4):512--530, 2001.

\bibitem[Jer92]{Jerrum92}
Mark Jerrum.
\newblock Large cliques elude the metropolis process.
\newblock {\em Random Struct. Algorithms}, 3(4):347--360, 1992.

\bibitem[Joh87]{Joh87}
David~S. Johnson.
\newblock The {NP}-completeness column: An ongoing guide.
\newblock {\em J. Algorithms}, 8(5):438--448, September 1987.

\bibitem[Kar72]{Karp72}
Richard~M. Karp.
\newblock Reducibility among combinatorial problems.
\newblock In {\em Proceedings of a symposium on the Complexity of Computer
  Computations}, pages 85--103, 1972.

\bibitem[Kar76]{Karp76}
Richard Karp.
\newblock Probabilistic analysis of some combinatorial search problems.
\newblock {\em Algorithms and Complexity: New Directions and Recent Results},
  1976.

\bibitem[Kho06]{Khot06}
Subhash Khot.
\newblock Ruling out {PTAS} for graph min-bisection, dense k-subgraph, and
  bipartite clique.
\newblock {\em {SIAM} J. Comput.}, 36(4):1025--1071, 2006.

\bibitem[KLL13]{KLL13}
Miroslaw Kowaluk, Andrzej Lingas, and Eva{-}Marta Lundell.
\newblock Counting and detecting small subgraphs via equations.
\newblock {\em {SIAM} J. Discret. Math.}, 27(2):892--909, 2013.

\bibitem[KP93]{KP93}
Guy Kortsarz and David Peleg.
\newblock On choosing a dense subgraph (extended abstract).
\newblock In {\em 34th Annual Symposium on Foundations of Computer Science,
  Palo Alto, California, USA, 3-5 November 1993}, pages 692--701, 1993.

\bibitem[KST54]{kHovari1954problem}
Tam{\'a}s K{\H{o}}v{\'a}ri, Vera~T S{\'o}s, and P{\'a}l Tur{\'a}n.
\newblock On a problem of {Zarankiewicz}.
\newblock In {\em Colloquium Mathematicum}, volume~3, pages 50--57, 1954.

\bibitem[Lin15]{Lin15}
Bingkai Lin.
\newblock The parameterized complexity of \emph{k}-biclique.
\newblock In {\em SODA}, pages 605--615, 2015.

\bibitem[LMS11]{LMS11}
Daniel Lokshtanov, D{\'{a}}niel Marx, and Saket Saurabh.
\newblock Lower bounds based on the exponential time hypothesis.
\newblock {\em Bull. {EATCS}}, 105:41--72, 2011.

\bibitem[LRS15]{LeeRS15}
James~R. Lee, Prasad Raghavendra, and David Steurer.
\newblock Lower bounds on the size of semidefinite programming relaxations.
\newblock In {\em STOC}, pages 567--576, 2015.

\bibitem[LRSZ17]{LRSZ17}
Daniel Lokshtanov, M.~S. Ramanujan, Saket Saurabh, and Meirav Zehavi.
\newblock Parameterized complexity and approximability of directed odd cycle
  transversal.
\newblock {\em CoRR}, abs/1704.04249, 2017.

\bibitem[LV20]{LV20}
Andrea Lincoln and Nikhil Vyas.
\newblock Algorithms and lower bounds for cycles and walks: Small space and
  sparse graphs.
\newblock In {\em ITCS}, pages 11:1--11:17, 2020.

\bibitem[Man17a]{Manurangsi17}
Pasin Manurangsi.
\newblock Almost-polynomial ratio {ETH}-hardness of approximating densest
  k-subgraph.
\newblock In {\em STOC}, pages 954--961, 2017.

\bibitem[Man17b]{Man17-ICALP}
Pasin Manurangsi.
\newblock Inapproximability of maximum edge biclique, maximum balanced biclique
  and minimum $k$-cut from the small set expansion hypothesis.
\newblock In {\em ICALP}, pages 79:1--79:14, 2017.

\bibitem[MM15]{MM15}
Pasin Manurangsi and Dana Moshkovitz.
\newblock Approximating dense max 2-{CSP}s.
\newblock In {\em APPROX}, pages 396--415, 2015.

\bibitem[MR17]{MR17}
Pasin Manurangsi and Aviad Rubinstein.
\newblock Inapproximability of {VC} dimension and {Littlestone's} dimension.
\newblock In {\em COLT}, pages 1432--1460, 2017.

\bibitem[MSOI{\etalchar{+}}02]{Milo02}
R.~Milo, S.~Shen-Orr, S.~Itzkovitz, N.~Kashtan, D.~Chklovskii, and U.~Alon.
\newblock Network motifs: Simple building blocks of complex networks.
\newblock {\em Science}, 298(5594):824--827, 2002.

\bibitem[NP85]{NP85}
Jaroslav Nesetril and Svatopluk Poljak.
\newblock On the complexity of the subgraph problem.
\newblock {\em Commentationes Mathematicae Universitatis Carolinae},
  26:415--419, 1985.

\bibitem[RS10]{RaghavendraS10}
Prasad Raghavendra and David Steurer.
\newblock Graph expansion and the unique games conjecture.
\newblock In {\em STOC}, pages 755--764, 2010.

\bibitem[Rub16]{Rubinstein16}
Aviad Rubinstein.
\newblock Settling the complexity of computing approximate two-player {Nash}
  equilibria.
\newblock In {\em FOCS}, pages 258--265, 2016.

\bibitem[Rub17]{Rubinstein17}
Aviad Rubinstein.
\newblock Detecting communities is hard (and counting them is even harder).
\newblock In {\em ITCS}, pages 42:1--42:13, 2017.

\bibitem[SS10]{SegevS10}
Danny Segev and Gil Segev.
\newblock Approximate \emph{k}-steiner forests via the lagrangian relaxation
  technique with internal preprocessing.
\newblock {\em Algorithmica}, 56(4):529--549, 2010.

\bibitem[UBK13]{UBK13}
Johan Ugander, Lars Backstrom, and Jon~M. Kleinberg.
\newblock Subgraph frequencies: mapping the empirical and extremal geography of
  large graph collections.
\newblock In {\em WWW}, pages 1307--1318, 2013.

\bibitem[Wil18]{Williams18}
Virginia~Vassilevska Williams.
\newblock On some fine-grained questions in algorithms and complexity, 2018.
\newblock ICM survey.

\bibitem[WW13]{WW13}
Virginia~Vassilevska Williams and Ryan Williams.
\newblock Finding, minimizing, and counting weighted subgraphs.
\newblock {\em {SIAM} J. Comput.}, 42(3):831--854, 2013.

\bibitem[YZ04]{YZ04}
Raphael Yuster and Uri Zwick.
\newblock Detecting short directed cycles using rectangular matrix
  multiplication and dynamic programming.
\newblock In {\em SODA}, pages 254--260, 2004.

\bibitem[Zuc96]{zuckerman1996unapproximable}
David Zuckerman.
\newblock On unapproximable versions of {NP}-complete problems.
\newblock {\em SIAM Journal on Computing}, 25(6):1293--1304, 1996.

\end{thebibliography}

\appendix

\section{Deferred proofs from Section~\ref{subsec:generic-soundness}}
\label{app:soundness-helper-lemmas}
In this appendix we include omitted proofs from Section~\ref{subsec:generic-soundness}.

\begin{obs*}[Restatement of Observation~\ref{obs:density-min-deg}]
For any $H = (V_H, G_H)$, there exists $S' \subseteq H$ such that $\mindeg(H[S']) \geq \den(H)$.
\end{obs*}

\begin{proof}[Proof of Observation~\ref{obs:density-min-deg}]
Let us repeatedly apply the following: remove any vertex with degree less than $\den(H)$ from the graph. In total, the number of edges removed is less than $\den(H) \cdot |V_H| = |E_H|$. Thus, the final graph is not empty; by definition, its minimum degree is at least $\den(H)$.
\end{proof}

\restatelemma{lem:disperser}

\begin{proof}[Proof of Lemma~\ref{lem:disperser}]
For a fixed $M \subseteq N$, we have
\begin{align*}
\Pr\left[\left|\bigcup_{i \in M} S_i\right| < 0.01 \delta |M| \ell\right]
&\leq \sum_{U \subseteq V \atop |U| = \lceil 0.01 \delta |M| \ell \rceil - 1} \Pr\left[\left(\bigcup_{i \in M} S_i\right) \subseteq U\right] \\
&= \sum_{U \subseteq V \atop |U| = \lceil 0.01 \delta |M| \ell \rceil - 1} \prod_{i \in M} \Pr\left[S_i \subseteq U\right] \\
&= \sum_{U \subseteq V \atop |U| = \lceil 0.01 \delta |M| \ell \rceil - 1} \prod_{i \in M} \left(|U| / n\right)^{\ell} \\
&= \sum_{U \subseteq V \atop |U| = \lceil 0.01 \delta |M| \ell \rceil - 1} \left(|U| / n\right)^{|M| \ell} \\
&\leq n^{\lceil 0.01 \delta |M| \ell \rceil - 1} \left(\frac{\lceil 0.01 \delta |M| \ell \rceil - 1}{n}\right)^{|M| \ell} \\
&\leq n^{0.01 \delta |M| \ell} \left(\frac{0.01 \delta |M| \ell}{n}\right)^{|M| \ell} \\
&= \left(\frac{0.01\delta|M| \ell}{n^{1 - 0.01\delta}}\right)^{|M|\ell}.
\end{align*}
Let $\Delta = \lfloor n^{0.99\delta} / \ell \rfloor$. Now, we take the union bound over all $M \subseteq [N]$ of size at most $\Delta$. We have
\begin{align*}
\Pr\left[\exists M \subseteq [N], |M| \leq \Delta \text{ and } \left|\bigcup_{i \in M} S_i\right| < 0.01 |M| \ell\right]
&\leq \sum_{t = 1}^{\Delta} \binom{N}{t} \left(\frac{0.01\delta t \ell}{n^{1 - 0.01\delta}}\right)^{t \ell} \\
&\leq \sum_{t = 1}^{\Delta} N^t \left(\frac{0.01\delta t \ell}{n^{1 - 0.01\delta}}\right)^{t \ell} \\
&\leq \sum_{t = 1}^{\Delta} \left(\frac{0.01\delta N^{1/\ell} t \ell}{n^{1 - 0.01\delta}}\right)^{t \ell} 
\intertext{ And since by assumption $N \leq 1000k \cdot n^{(1 - \delta) \ell}$ and $\ell \geq k \geq 20$,} 
&\leq \sum_{t = 1}^{\Delta} \left(\frac{0.04 t \ell}{n^{0.99\delta}}\right)^{t \ell} \\
\intertext{And finally because we have $t \leq \Delta \leq n^{0.99\delta} / \ell$,} 
&\leq \sum_{t = 1}^{\Delta} 0.04^t
\, \leq 0.05. \qedhere
\end{align*}
\end{proof}

\end{document}